\documentclass[a4paper,11pt]{article}
\usepackage[margin=0pt]{geometry} 

\usepackage{pgfplots}
\pgfplotsset{compat=1.15}
\usepackage{mathrsfs}
\usetikzlibrary{arrows}

\usepackage[T1]{fontenc}

\usepackage[]{amsmath, amssymb, amsthm, tabularx}

\usepackage[]{a4, xcolor, here}

\usepackage[]{pstricks, pst-text, pst-node, pst-tree, gastex}

\usepackage{tikz}

\usepackage[tikz]{bclogo}

\usepackage{comment}

\usepackage{boites,graphicx}

\usepackage{soul}

\definecolor{pastelyellow}{rgb}{0.99, 0.99, 0.59}
\definecolor{aqua}{rgb}{0.0, 1.0, 1.0} 
\definecolor{aquamarine}{rgb}{0.5, 1.0, 0.83} 
\definecolor{bananayellow}{rgb}{1.0, 0.88, 0.21}
\definecolor{burgundy}{rgb}{0.5, 0.0, 0.13}
\definecolor{ao(english)}{rgb}{0.0, 0.5, 0.0}

\setul{}{0.2ex}
\setulcolor{bananayellow}

\usepackage[normalem]{ulem}

\usepackage{mathrsfs}

\usepackage{stmaryrd}

\usepackage{enumerate}

\usepackage{paralist}

\usepackage{multienum}

\usepackage{multicol}

\usepackage{fancyhdr}

\usepackage{datetime}

\usepackage[]{hyperref} 
\hypersetup{
	colorlinks = true,
	linkcolor = red,
	anchorcolor = black,
	citecolor = blue, 
	filecolor = cyan,
	menucolor = red,
	runcolor = cyan,
	urlcolor = blue,
	linkbordercolor = {white},
	linktocpage = true
}

\newtheorem{theorem}{Theorem}[section]
\newtheorem{proposition}[theorem]{Proposition}

\newtheorem{corollary}[theorem]{Corollary}

\theoremstyle{definition}
\newtheorem{definition}[theorem]{Definition}

\newtheorem{example}[theorem]{Example}
\newtheorem{remark}[theorem]{Remark}

\makeatletter
\def\thmhead@plain#1#2#3{%
	\thmname{#1}\thmnumber{\@ifnotempty{#1}{ }\@upn{#2}}%
	\thmnote{ {\the\thm@notefont#3}}}
\let\thmhead\thmhead@plain
\makeatother

\newcommand{\cC}{\mathcal{C}}

\newcommand{\cF}{\mathcal{F}}
\newcommand{\cG}{\mathcal{G}}

\newcommand{\cP}{\mathcal{P}}

\newcommand{\cU}{\mathcal{U}}
\newcommand{\cV}{\mathcal{V}}

\newcommand{\cX}{\mathcal{X}}

\newcommand{\rsp}[1]{{\mathrm{rowsp}{#1}}}

\newcommand{\stab}{\mathrm{Stab}}
\newcommand{\stabbeta}{\mathrm{Stab}_{\beta}}

\newcommand{\orb}{\mathrm{Orb}}
\newcommand{\orbbeta}{\mathrm{Orb}_{\beta}}
\newcommand{\GL}{\mathrm{GL}}

\newcommand{\bbF}{{\mathbb F}} 

\renewcommand{\geq}{\geqslant}
\renewcommand{\leq}{\leqslant}

	{\end{list}}%

\begin{document}

	\renewcommand{\headrulewidth}{0pt}
	
	\rhead{ }
	\chead{\scriptsize  On Generalized Galois Cyclic Orbit Flag Codes}
	\lhead{ }

	\title{On Generalized Galois Cyclic Orbit Flag Codes}

	\author{\renewcommand\thefootnote{\arabic{footnote}}
		Clementa Alonso-Gonz\'alez\footnotemark[1] \ and  Miguel \'Angel Navarro-P\'erez\footnotemark[1]}

	\footnotetext[1]{Dpt.\ de Matem\`atiques, Universitat d'Alacant, 
		Sant Vicent del Raspeig, Ap.\ Correus 99, E -- 03080 Alacant.}

	{\small \date{\usdate{\today}}} 
	
	\maketitle

	\begin{abstract}
		Flag codes that are orbits of a cyclic subgroup of the general linear group acting on flags of a vector space over a finite field, are called {\em cyclic orbit flag codes}. In this paper we present a new contribution to the study of such codes started in \cite{FlagCyclic}, by focusing this time on the generating flag. More precisely, we examine those ones whose generating flag has at least one subfield among its subspaces. In this situation, two important families arise: the already known {\em Galois flag codes}, in case we have just fields, or the {\em generalized Galois flag codes} in other case. We investigate the parameters and properties of the latter ones and explore the relationship with their underlying Galois flag code. 
	\end{abstract}
	


	\section{Introduction}\label{sec:Introduction}
	
	Network coding represents a procedure to data transfer within a network that is a directed multigraph without cycles, where the information travels from one or several senders to several receivers. In \cite{AhlsCai00}, it was proved that one can improve the efficiency if the intermediate nodes can linearly combine the information vectors. We speak about \emph{random network coding} whenever the underlying network topology is unknown. Due to the fact that vector subspaces are invariant under linear combinations, they are proposed as suitable codewords in \cite{KoetKschi08}, giving raise to the concept of  \emph{subspace codes}. When all the subspaces have a fixed dimension, we get \emph{constant dimension codes}. Research in this area has been very intense in latter years (consult \cite{TrautRosen18} and references inside).
	
	One method devised in \cite{TrautManRos2010} to build subspace codes consists of making subgroups of the general linear group $\GL(n,q)$ act on the set of subspaces of $\bbF_q^n$ and thus, consider the corresponding orbits. This idea leads to the concept of \emph{orbit codes}. In particular, when the acting group is cyclic, we obtain the so-called \emph{cyclic orbit codes}, widely studied in the last times (see \cite{BenEtGaRa16, ChenLi18, EtVar11, GLMoTro2015, OtOz17, RothRaTa18, TrautManBraunRos2013, TrautManRos2010, ZhaoTang2019}, for instance). Of special relevance for our purposes is the paper \cite{GLMoTro2015}, where the authors treat $\beta$-\emph{cyclic orbit codes} as collections of $\bbF_q$-vector subspaces of $\bbF_{q^n}$ that are orbits under the natural action of a subgroup $\langle \beta \rangle$ of $\bbF_{q^n}^\ast$ on $\bbF_q$-vector spaces (if $\beta$ is primitive, the corresponding orbit is called just \emph{cyclic orbit code}).  In that paper, it is introduced an interesting tool to analyze $\beta$-cyclic orbit codes: the \emph{best friend} of the code, that is, the largest subfield of $\bbF_{q^n}$ over which the \emph{generating subspace} is a vector space.
	
	\emph{Flag codes} can be seen as a generalization of constant dimension codes. The codewords of a flag code are flags, that is, sequences of chained subspaces of prescribed dimensions. In the network coding context, they appeared for the first time in the paper \cite{LiebNebeVaz18} where the multiplicative action of $\GL(n,q)$ is translated from subspaces to flags to provide different  constructions of \emph{orbit flag codes}. This seminal work has sparked an incipient interest in flag codes reflected in the recent works \cite{Consistentes, FlagCyclic, OrbitODFC, cotas, CasoPlanar,CasoNoPlanar,  Kurz, MA-X}.
	
	In \cite{FlagCyclic},  the authors undertake the study of $\beta$-\emph{cyclic orbit flag codes} inspired by the ideas in \cite{GLMoTro2015}. More precisely, they consider flags on $\bbF_{q^n}$ given by nested $\bbF_q$-subspaces of the field $\bbF_{q^n}$ constructed as orbits of subgroups $\langle \beta \rangle \subseteq \bbF_{q^n}^\ast$, and coin the concept of \emph{best friend} of a cyclic flag code as the largest subfield of $\bbF_{q^n}$ over which every subspace in the \emph{generating flag} is a vector space. The knowledge of the best friend turns out to be extremely useful to determine the parameters of the code as well as other features such as the necessary conditions on the type vector to reach the maximum distance. It is also presented the particular family of \emph{Galois flag codes}, that consists of $\beta$-cyclic orbit flag codes generated by flags given by nested fields. For that class of codes it is possible to precisely establish a nice correspondence between the set of attainable distances and the subgroups of $\bbF_{q^n}^*$. 
	
	In the current paper we extend the study performed in \cite{FlagCyclic} by focusing on the generating flag. More precisely, we examine $\beta$-cyclic orbit flag codes whose generating flag has at least one subfield among its subspaces. We distinguish two situations: either all the subspaces are fields, then we have the already known \emph{Galois flag codes}, or there is also at least one subspace not being a field. The last case entails the definition of a new kind of $\beta$-cyclic orbit flag codes called \emph{generalized Galois flag codes}. We discuss the properties of this new class of codes by taking into account that a generalized Galois flag code has always an underlying Galois flag code that influences on it to a greater or lesser extent. 
	
	The text is structured as follows. In section \ref{sec:Preliminaries} we remember the basics on subspace codes as well as some notions and results related to cyclic orbit (subspace) codes developed in \cite{GLMoTro2015}. In Section \ref{sec: cyclic orbit flag codes} we recall, on the one hand, some background on flag codes and the most important facts on cyclic orbit flag codes that appear in \cite{FlagCyclic}. On the other hand, we present some new results on the interplay between type vectors, best friend and the flag distance parameter. The family of generalized $\beta$-Galois flag codes is introduced here as an extension of the $\beta$-Galois flag codes. We discuss in which way the properties of a generalized Galois flag code are driven by its underlying Galois flag code and launch a related question. Section \ref{sec: A construction} is devoted to provide a systematic construction of generalized $\beta$-Galois flag codes with a prescribed underlying $\beta$-Galois flag code by using generating flags written in a precise regular form.  
	In Subsections \ref{subsec: basic}, \ref{subsec:weaved} and \ref{subsec: decoding} we analyze the particular properties of the previous construction in case $\beta$ is primitive, and propose a decoding algorithm over the erasure channel taking advantage of such properties. In Subsection \ref{subsec:weaved beta} we address the case when $\beta$ is not primitive and present some specific results. To finish, we show how our construction allows us to give an answer to the question previously formulated. 
	

	\section{Preliminaries}\label{sec:Preliminaries}
	Consider $q$ a prime power and $\bbF_q$ the finite field with $q$ elements. We denote by $\bbF_q^n$ the $n$-dimensional vector space over $\bbF_q$ for any natural number $n\geq 1$ and by $\cP_q(n)$ the set of all the subspaces of $\bbF_q^n$. For every $0\leq k \leq n$, the set of $k$-dimensional subspaces of $\bbF_q^n$, that is, the \emph{Grassmannian}, will be denoted by $\cG_q(k, n)$. 
	The set $\cP_q(n)$ can be equipped with a metric called the \emph{subspace distance}: for any pair $\cU, \cV \in \cP_q(n)$, we set 
	\begin{equation}\label{def: subspace distance}
		d_S(\cU, \cV)= \dim(\cU+\cV)-\dim(\cU\cap\cV).
	\end{equation}
	In particular, the subspace distance between two subspaces $\cU, \cV \in \cG_q(k, n)$ becomes
	\begin{equation}\label{def: subspace distance in the Grassmannian}
		d_S(\cU, \cV)= 2(k-\dim(\cU\cap\cV)).
	\end{equation}
	A \emph{constant dimension code} $\cC$ of dimension $k$ and length $n$ is a nonempty subset of $\cG_q(k, n)$ whose \emph{minimum subspace distance} is given by
	$$
	d_S(\cC)=\min\{ d_S(\cU, \cV) \ | \ \cU, \cV \in \cC, \ \cU \neq \cV \}.
	$$
	If $|\cC|=1$, we put $d_S(\cC)=0$. For further details on this family of codes, consult \cite{TrautRosen18} and the references inside. 
	
	It is clear that the minimum distance of a constant dimension code $\cC \subseteq \cG_q(k, n)$ is attained when the intersection of every pair of codewords has the minimum possible dimension. In this case, we have that 
	\begin{equation}\label{eq: bound subspace distance}
		d_S(\cC)\leq
		\left\lbrace
		\begin{array}{lll}
			2k      & \text{if} & 2k\leq n, \\
			2(n-k)  & \text{if} & 2k > n.
		\end{array}
		\right.
	\end{equation}
	A constant dimension code with dimension $k \leq \lfloor\frac{n}{2}\rfloor$ attaining the previous bound is called \emph{partial spread code}. A partial spread code being also a partition of $\bbF_q^n$ into $k$-dimensional subspaces is known as a \emph{spread code} or just a \emph{$k$-spread}.  In \cite{Segre64} it is proved that $k$-spread exist if, and only if, $k$ divides $n$. As a result, the size of any $k$-spread is exactly $\frac{q^n-1}{q^k-1}.$ See \cite{GoManRo12, MangaGorlaRosen08, MangaTraut14, TrautRosen18} for more information concerning spread codes in the network coding setting.

	Among all the special families of constant dimension codes, here we are interested in \emph{orbit codes}, that is, those that arise as orbits of the action of subgroups of the general linear group $\GL(n, q)$ on the Grassmannian. This family of codes was introduced in \cite{TrautManRos2010}. More precisely, fixed a $k$-dimensional subspace $\cU \subset \bbF_q^n$ and a subgroup $G \subseteq \GL(n, q)$, the orbit of $\cU$ under the action of $G$ is the constant dimension code given by $\mathrm{Orb}_G(\cU) = \{ \cU \cdot A \ | \ A\in G\},$
	where $\cU\cdot A = \rsp (UA),$  for any full-rank generator matrix $U$ of $\cU$. The \emph{stabilizer} of $\cU$ under the action of $G$ is the subgroup $\mathrm{Stab}_G(\cU)= \{ A\in G  \ | \ \cU\cdot A= \cU \}.$
	As a consequence,
	\begin{equation}
		| \mathrm{Orb}_G(\cU)| = \frac{|G|}{|\mathrm{Stab}_G(\cU)|}
	\end{equation}
	and the minimum distance can be computed as
	$$
	d_S( \mathrm{Orb}_G(\cU))= \min \{ d_s(\cU, \cU\cdot A) \ | \ A \in G \setminus\mathrm{Stab}_G(\cU) \}.
	$$
	Whenever the acting group $G$ is cyclic, the orbit $\mathrm{Orb}_G(\cU)$ is called \emph{cyclic orbit code.} The works \cite{GLMoTro2015, ManTrautRos11, RosTraut2013, TrautManBraunRos2013} are devoted to the study of this family of codes. In the current paper we are specially interested in the viewpoint developed in \cite{TrautManBraunRos2013} and \cite{GLMoTro2015} where, taking advantage of the natural $\bbF_q$-linear isomorphism between $\bbF_q^n$ and $\bbF_{q^n}$, cyclic orbit codes are seen as collections of subspaces in $\bbF_{q^n}$. More precisely, in \cite{GLMoTro2015} the authors consider a nonzero element $\beta$ and define $\beta$-\emph{cyclic orbit codes} as orbits of the group $\langle \beta \rangle$ on $\bbF_q$-vector subspaces of $\bbF_{q^n}.$ In particular, if $1\leq k<n$ and $\cU\subset \bbF_{q^n}$ is a $k$-dimensional subspace over $\bbF_q$, the $\beta$-\emph{cyclic orbit code generated by $\cU$} is the following set of $\bbF_q$-subspaces of dimension $k$ 
	$$\orbbeta(\cU)= \{\cU \beta^i \ | \ 0\leq i \leq |\beta|-1\},$$
	where $|\beta|$ denotes the \emph{multiplicative order} of $\beta.$ The \emph{stabilizer} of the subspace $\cU$ under the action of $\langle \beta \rangle$ is the  cyclic subgroup $\stabbeta(\cU)= \{\beta^i \ | \ \cU\beta^i= \cU \}.$ An important example of such kind of codes, already developed in \cite{TrautManBraunRos2013}, is the following $k$-spread code, where $k$ is a divisor of $n$ and $\alpha$ is a primitive element of  $\bbF_{q^n}$:
	\begin{equation}\label{def: spread Anna-Lena}
		 \orb_{\langle \alpha \rangle}(\bbF_{q^k})= \{ \bbF_{q^k}\alpha^i \ | \ i=0, \ldots, q^n-2 \}.
	\end{equation}
	\begin{remark}
		Following the notation used in \cite{GLMoTro2015}, when the acting group is $\bbF_{q^n}^\ast$, we simply denote the corresponding orbit by $\orb(\cU)$ and call it just the \emph{cyclic orbit code generated by $\cU$}. In this situation, we also remove the subscript $\beta$ and write $\stab(\cU)$ to denote the stabilizer of $\cU$.
	\end{remark}
	
	Concerning the cardinality and distance of a $\beta$-cyclic orbit code, in \cite{GLMoTro2015} the authors study these parameters with the aid of the \emph{best friend} of the generating subspace. This concept is closely linked to the stabilizer of the subspace. Let us recall the definition. 
	\begin{definition}
		A subfield $\bbF_{q^m}$ of $\bbF_{q^n}$ is said to be a \emph{friend} of a subspace $\cU \subset \bbF_{q^n} $ if $\cU$ is an $\bbF_{q^m}$-vector space. The largest friend of $\cU$ is called its \emph{best friend}. 
	\end{definition}
	
	The knowledge of the best friend of a subspace $\cU$ provides straightforwardly the cardinality of the cyclic orbit code as well as a lower bound for its distance.
	
	\begin{proposition}\label{prop: stab+ es best friend}(\cite[ Prop. 3.3, 3.12, 3.13 and 4.1]{GLMoTro2015})
		Let $\cU$ be a subspace of $\bbF_{q^n}$ with the subfield $\bbF_{q^m}$ as its best friend. Then
		$$
		|\orb(\cU)|=\frac{q^n-1}{q^m-1}.
		$$
		Moreover, the value $2m$ divides the distance between every pair of subspaces in $\orb(\cU)$ and, hence, we have that  $d_S(\orb(\cU))\geq 2m.$
	\end{proposition}

	We finish this section recalling a construction of cyclic orbit codes with prescribed distance and cardinality from the choice of a subspace $\cU$ written in a specific regular form.
	
	\begin{proposition}\label{prop: suma directa}(\cite[Prop. 4.3]{GLMoTro2015})
		Consider the subspace $\cU = \bigoplus_{i=0}^{t-1} \bbF_{q^m}\alpha^{li}$ for some $1\leq l <\frac{q^n-1}{q^m-1}$ such that $\mathbb{F}_{q^m}$ is the best friend of $\cU$. Then $d_S(\orb(\cU))=2m$.
	\end{proposition}
	
	It is clear that the subfield $\bbF_{q^m}$ is a friend of a subspace $\cU$ written as in previous proposition, although it is not necessarily its best friend. In fact, there are just two possibilities for the best friend of $\cU$.
	
	\begin{proposition}\label{prop: BF en forma regular}(\cite[Proposition 4.4]{GLMoTro2015})
		Given the subspace $\cU = \bigoplus_{i=0}^{t-1} \bbF_{q^m}\alpha^{li}$ for some $1\leq l <\frac{q^n-1}{q^m-1}.$ If $f(x)$ is the minimal polynomial of $\alpha^l$ over $\bbF_{q^m}$, then its degree is at least $t$ and
		$$
		\cU=\bbF_{q^{mt}} \ \Leftrightarrow \ \deg(f)=t \ \Leftrightarrow \ \alpha^l\in \stab(\cU) \ \Leftrightarrow \ \bbF_{q^m}  \text{ is not the best friend of} \ \cU.
		$$
	\end{proposition}
	
	We will come back to this family of subspaces and to the $\beta$-cyclic orbit codes generated by them in Section \ref{sec: A construction}, where we provide a specific construction of generalized Galois cyclic orbit flag codes by using subspaces written as in Proposition \ref{prop: suma directa}.

	\section{Cyclic orbit flag codes}\label{sec: cyclic orbit flag codes}
	Part of this section is dedicated to gather the basic background on flag codes that already appears in \cite{ CasoPlanar, CasoNoPlanar,LiebNebeVaz18}, and to recall the main definitions and results that pertain to the particular class of cyclic orbit flag codes introduced in \cite{FlagCyclic}. In Subsection \ref{subsec:interdependence}, we present new results concerning the interdependence between the minimum distance of a $\beta$-cyclic orbit flag code, its best friend and the set of dimensions appearing in the type vector. In addition, the class of generalized Galois flag codes is introduced in Subsection \ref{subsec: generalized Galois}.

	\subsection{Flags and flag codes}

	A {\em flag} $\mathcal{F}=(\mathcal{F}_1,\ldots,  \mathcal{F}_r)$ on the extension field $\mathbb{F}_{q^n}$ is a sequence of nested $\bbF_q$-vector subspaces 
	$$
	\{0\}\subsetneq \mathcal{F}_1 \subsetneq \cdots \subsetneq \mathcal{F}_r \subsetneq \mathbb{F}_{q^n}.
	$$
	The subspace $\mathcal{F}_i$ is called the {\em $i$-th subspace} of $\cF$ and the {\em type} of $\mathcal{F}$ is the vector $(\dim(\cF_1), \dots, \dim(\cF_r))$. When the type vector is $(1, 2, \ldots, n-1),$ we say that ${\cF}$ is a {\em full flag}. Given two different flags $\cF$, $\cF'$ on  $\mathbb{F}_{q^n}$, we say that $\cF'$ is a \emph{subflag} of $\cF$ if each subspace of $\cF'$ is a also subspace of $\cF$.
	
	The \emph{flag variety} of type $(t_1, \dots, t_r)$ on $\mathbb{F}_{q^n}$ is the set of flags of this type and will be denoted by $\cF_q((t_1, \dots, t_r),n)$. Note that  $\cF_q((t_1, \dots, t_r),n)$ embeds in the product of Grassmannians $\cG_q(t_1, n) \times \cdots \times \cG_q(t_r, n)$ and, hence, this variety can be endowed with a metric that extends the subspace distance defined in (\ref{def: subspace distance}). Given two flags $\cF=(\mathcal{F}_1,\ldots,  \mathcal{F}_r)$ and $\cF'=(\mathcal{F}'_1,\ldots,  \mathcal{F}'_r)$ in $\mathcal{F}_q( (t_1, \ldots, t_r),n)$, their \emph{flag distance} is
	$$
	d_f(\cF,\cF')= \sum_{i=1}^r d_S(\mathcal{F}_i, \mathcal{F}'_i).
	$$
	\begin{definition}
		A \emph{flag code} $\cC$ of type $(t_1, \dots, t_r)$ on $\bbF_{q^n}$ is a nonempty subset of $ \cF_q((t_1, \dots, t_r), n)$. The {\em minimum distance} of $\cC$ is given by
		$$
		d_f(\cC)=\min\{d_f(\cF,\cF')\ |\ \cF,\cF'\in\cC, \ \cF\neq \cF'\}.
		$$
	\end{definition}
	\noindent whenever $\cC$ has more that two elements. In case $|\cC|=1$, we put $d_f(\cC)=0.$ For type $(t_1, \dots, t_r),$ it always holds
	\begin{equation}\label{eq: dist max flags}
		d_f(\cC) \leq	2 \left( \sum_{t_i \leq \lfloor \frac{n}{2}\rfloor} t_i + \sum_{t_i > \lfloor \frac{n}{2}\rfloor} (n-t_i) \right).
	\end{equation}
	
	There are constant dimension codes intrinsically correlated with a flag code $\cC$ that play an important role in the study of parameters and properties of $\cC$.
	
	\begin{definition}
		Given a flag code $\cC$ of type $(t_1, \dots, t_r)$, the \emph{$i$-projected code} of $\cC$ is the set  
		$$
		\cC_i=\{\cF_i\ |\ (\cF_{1}, \dots, \cF_{i}, \dots, \cF_{r}) \in \cC \} \subseteq \cG(t_i, n).
		$$
	\end{definition}
	
	\begin{remark}
		Concerning the relationship between the size of a flag code and the ones of its projected codes, it is clear that  $\vert \cC_i\vert \leq \vert \cC \vert$ for every $i=1, \dots, r$. In case  $|\cC_1|=\dots=|\cC_r|=|\cC|$, we say that  $\cC$ is \emph{disjoint}. Under this condition, it is possible to establish also a clear connection between the minimum distance of a given flag code and the ones of its projected codes.  More precisely, if $\cC$ is a disjoint flag code, then
		$$
		d_f(\cC) \geq \sum_{i=1}^r d_S(\cC_i).
		$$
	\end{remark}
	
	In \cite{Consistentes} the authors introduced a family of flag codes such that the distance and size of the projected codes completely determine the ones of the corresponding flag code. 
	\begin{definition}\label{def: consistentes}
		A flag code $\cC$ is \emph{consistent} if the following conditions hold:
		\begin{enumerate}
			\item  $\cC$ is disjoint.
			\item  $d_f(\cC) =\sum_{i=1}^r d_S(\cC_i).$
		\end{enumerate}
	\end{definition}
	
	In the same paper, the authors develop a decoding algorithm for consistent flag codes over the erasure channel and provide important families of such a class of codes. Among them, we can find the one of \emph{optimum distance flag codes}. This class of flag codes has been already studied in \cite{OrbitODFC, CasoPlanar, CasoNoPlanar, MA-X}.  In these works, the reader can find specific constructions of them as well as the following characterization.
	
	\begin{theorem}\cite[Th. 3.11]{CasoPlanar}\label{theo: caracterización ODFC}
		A flag code is an optimum distance flag code if, and only if, it is disjoint and every projected code attains the maximum possible distance for its dimension.
	\end{theorem}

	We will come back to this ideas in Subsection \ref{subsec: decoding} in order to adapt the consistent flag codes decoding algorithm designed in \cite{Consistentes} to the constructions proposed in the present paper.

	\subsection{Cyclic orbit flag codes}
	
	Let us recall the concept of \emph{cyclic orbit flag code} as the orbit of the multiplicative action of (cyclic) subgroups of $\bbF_{q^n}^\ast$ on flags on $\bbF_{q^n}$. This concept of cyclic orbit flag code was first introduced in \cite{FlagCyclic} following the approach of \cite{GLMoTro2015} for cyclic orbit subspace codes. 
	
 The cyclic group $\bbF_{q^n}^\ast$ acts on flags on $\bbF_{q^n}$ as follows: given $\beta\in\bbF_{q^n}^\ast$ and a flag $\cF=(\mathcal{F}_1,\dots,  \mathcal{F}_r)$ of type $(t_1, \dots, t_r)$, the flag $\cF\beta$ is
	\begin{equation}\label{def: action of beta on flags}
		\cF\beta =  (\cF_1\beta, \ldots, \cF_r\beta)
	\end{equation}
and the orbit
\begin{equation}\label{def: cyclic orbit flag code}
		\orbbeta(\cF) = \{ \cF \beta^j \ | \ 0\leq j \leq |\beta|-1 \}
	\end{equation}
	is called the $\beta$-\emph{cyclic orbit flag code} generated by $\cF.$ The \emph{stabilizer} of  $\cF$ (w.r.t. $\beta$) is the subgroup of $\langle \beta \rangle$ given by
	\begin{equation}\label{def: stabilizer of a flag}
		\stabbeta(\cF)= \{\beta^j \ | \ \cF\beta^j= \cF \}.
	\end{equation}
	
	If $\beta$ is primitive, that is, if the acting group is $\bbF_{q^n}^\ast$, we  simply write $\orb(\cF)$ to denote the \emph{cyclic orbit flag code generated by $\cF$.} We also drop the subscript in $\stab(\cF)$. Observe that, for every $\beta\in\bbF_{q^n}^\ast$, it holds $\stabbeta(\cF)=\langle\beta\rangle \cap \stab(\cF)$.
	
	We can take advantage of the orbital structure to compute the code parameters: the cardinality of $ \orbbeta(\cF)$ is given by
	\begin{equation}\label{cardinality cyclic orbit flag code}
		| \orbbeta(\cF)| = \dfrac{|\beta|}{|\stabbeta(\cF)|}  = \dfrac{|\beta|}{|\langle\beta\rangle \cap \stab(\cF)|}
	\end{equation}
	and its minimum distance can be calculated as
	\begin{equation}\label{distance cyclic orbit flag code}
		d_f(\orbbeta(\cF))= \min\{ d_f(\cF, \cF\beta^j) \ | \ \beta^j \notin \stabbeta(\cF) \}.
	\end{equation}
	\begin{remark}
		Concerning the projected codes associated to  $\orbbeta(\cF)$, there are important facts to point out. First of all, note that the projected codes of a $\beta$-cyclic orbit flag codes are also $\beta$-cyclic orbit (subspace) codes. More precisely, for every $1\leq i\leq r$, we have
		\begin{equation}\label{def: prjected codes of an orbit flag code}
			(\orbbeta(\cF))_i = \orbbeta(\cF_i).
		\end{equation}
		Moreover, the straightforward stabilizers relationship
		\begin{equation}\label{eq: estabilizador flag}
			\stabbeta(\cF)=\bigcap_{i=1}^r \stabbeta(\cF_i) 
		\end{equation}
		leads to a nice rapport  between cardinalities: for every $1\leq i \leq r$, we have that $|\orbbeta(\cF_i)|$ divides $|\orbbeta(\cF)|$ (\cite[Prop. 3.6]{FlagCyclic}).
	\end{remark}
	
	Coming back to the computation of the values $| \orbbeta(\cF)|$ and $d_f(\orbbeta(\cF))$, in \cite{FlagCyclic}, it is showed that the knowledge of a specific subfield associated to $\cF$ allows us to obtain them directly. Let us recall the concept of best friend of a flag introduced in \cite{FlagCyclic} by generalization of the concept of a subspace best friend given in \cite{GLMoTro2015}.  
	
	\begin{definition}
		A subfield $\bbF_{q^m}$ of $\bbF_{q^n}$ is said to be a \emph{friend} of a flag $\cF$ on $\bbF_{q^n}$ if all the subspaces of $\cF$ are $\bbF_{q^m}$-vector spaces, that is, if it is a friend of all of them. The \emph{best friend} of the flag $\cF$ is its biggest friend.
	\end{definition} 
	
	From this definition it clearly holds that the type vector of a flag has to satisfy a necessary condition whenever the best friend is fixed. Furthermore, as it occurs when we work with the stabilizer subgroup, there are important connections between the best friend of a flag and the ones of its subspaces. 
	
	\begin{proposition}\cite[Lemma 3.14, Prop. 3.16, Cor. 3.18]{FlagCyclic} \label{lem: BF divides dimensions}
		Let $\cF=(\cF_1, \ldots, \cF_r)$ be a flag of type $(t_1, \ldots, t_r)$ on $\bbF_{q^n}$. If $\bbF_{q^m}$ is a friend of $\cF,$ then $m$ divides  $\gcd(t_1, \ldots, t_r, n).$ Moreover, if $\bbF_{q^m}$ is the best friend of $\cF$, then it is the intersection of the ones of $\cF_i$, for every $i=1,\ldots, r$, and we also have that $\bbF_{q^m}=\stab(\cF) \cup \{0\}$.
	\end{proposition}
	\begin{remark}
		Note that, if $1\in\cF_1$, then every friend of the flag $\cF$ is contained in $\cF_1$. Moreover, all the flags in $\orbbeta(\cF)$ have the same best friend, allowing us to speak about the \emph{best friend of a $\beta$-cyclic orbit flag code}. Now, if $\cF=(\cF_1, \ldots, \cF_r)$ is a flag of type $(t_1, \ldots, t_r)$ on $\bbF_{q^n}$ with  $\bbF_{q^m}$ as its best friend, $\bbF_{q^m}$ must be a friend of all its subspaces and we can write $t_i=m s_i$ for $i=1, \ldots, r$,  where $1\leq s_1 < \dots < s_r < s = \frac{n}{m}$. Finally, we can find linearly independent elements $a_1, \ldots, a_{s_r} \in \bbF_{q^{n}}$ (over $\bbF_{q^m}$) such that, for every $1\leq i\leq r$, we have
		\begin{equation}\label{eq: sum direct F_i}
			\cF_i = \bigoplus_{j=1}^{s_i}  \bbF_{q^m} a_{j}.
		\end{equation}
		In case $m$ is a dimension in the type vector, then $s_1=1$ and the cyclic orbit code $\orb(\cF_1)$ is the $m$-spread of $\bbF_{q^n}$ described in (\ref{def: spread Anna-Lena}). Moreover, if $1\in\cF_1$, this subspace is exactly the subfield $\bbF_{q^m}$.
		
	\end{remark}
	
	Let us recall how the knowledge of the best friend of a $\beta$-cyclic orbit flag code provides relevant information about the code parameters.
	
	\begin{proposition}\cite{FlagCyclic}
		\label{prop: cardinality and best friend}
		Let $\cF=(\cF_1, \ldots, \cF_r)$ be a flag on $\bbF_{q^n}$ and assume that $\bbF_{q^m}$ is its best friend. Then 
		\begin{equation}\label{eq: cardinality beta-ciclicos}
			|\orbbeta(\cF)|=\frac{|\beta|}{|\langle\beta\rangle \cap \bbF_{q^m}^\ast|}.
		\end{equation}
		Moreover, the value $2m$ divides $d_f(\orbbeta(\cF))$ and, if the the type vector of $\cF$ is $(ms_1, \ldots, ms_r)$, then it holds
		\begin{equation}\label{eq: distance bounds}
			2m \leq d_f(\orb(\cF))  \leq  2m \left( \sum_{s_i \leq \lfloor \frac{s}{2}\rfloor} s_i + \sum_{s_i > \lfloor \frac{s}{2}\rfloor} (s-s_i) \right),
		\end{equation}
		  whenever $\beta\notin\bbF_{q^m}^\ast$. On the other hand, if $\beta\in\bbF_{q^m}^\ast$, then $d_f(\orbbeta(\cF))=0$.
	\end{proposition}
	\begin{remark} From (\ref{eq: cardinality beta-ciclicos}) and (\ref{eq: distance bounds}), it is clear that both size and cardinality depend on the generating flag (hence on its best friend), the acting subgroup and the type vector. In particular, once we have fixed the best friend $\bbF_{q^m}$, we obtain the maximum possible orbit size if $\beta$ is a primitive element of $\bbF_{q^n}$. In this case, it holds $|\orb(\cF)|=\frac{q^n-1}{q^m-1}$. However, if we take $\beta\in\bbF_{q^m}^\ast$, we obtain the minimum possible cardinality since $\orbbeta(\cF)=\{\cF\}$.
	\end{remark}

	\subsection{Flag distances, best friend and type vectors interplay} \label{subsec:interdependence}
	
	From the bounds provided in (\ref{eq: distance bounds}) we know that, fixed the subfield $\bbF_{q^m}$ as best friend of a flag code $\cC$ of type $(ms_1, \ldots, ms_r)$, the possible values for the distance between flags in $\cC$ belong to the interval 
	\begin{equation}\label{eq: distances interval}
		[2m,2m \sum_{s_i \leq \lfloor \frac{s}{2}\rfloor} s_i + \sum_{s_i > \lfloor \frac{s}{2}\rfloor} (s-s_i) ].    
	\end{equation}		
	
	Nevertheless, in the orbital flag codes setting it is very important to point out that not every possible flag distance value is compatible with every type vector. In general, the greater the flag distance, the more conditions over the corresponding type vector we will have to impose.  The simplest case comes from considering cyclic flag codes of length one. In \cite[Lemma 4.1]{GLMoTro2015}, it was already shown that a cyclic (subspace) code with best friend $\bbF_{q^m}$ has, at least, distance $2m$ and constructions attaining this extreme value of the distance were also provided in \cite[Prop. 4.3]{GLMoTro2015}. However, when we work with flags of length $r\geq 2$, not even the minimum value of the distance, which is $2m$ as well, can be obtained for every type vector. This is a consequence of the link between flag distance values and the number of subspaces of a flag $\cF$ that share the best friend of $\cF$. Let us explain this relationship in the following result.

	\begin{theorem}\label{th: dist min type vector}
		Let $\cF$ be a flag on $\bbF_{q^n}$ with the subfield $\bbF_{q^m}$ as its best friend and take $\beta\in\bbF_{q^n}^\ast\setminus\bbF_{q^m}^\ast$. 
		\begin{enumerate}
			\item If there are $1\leq j\leq r$ subspaces of $\cF$ with $\bbF_{q^m}$ as their best friend, then 
			$d_f(\orbbeta(\cF))\geq 2mj.$
			\item If $d_f(\orbbeta(\cF))=2m$, then $\bbF_{q^m}$ is the best friend of exactly one subspace of $\cF$.
		\end{enumerate}
	\end{theorem}
	\begin{proof} Let us prove $(1)$. Assume that there exist $j$ subspaces, say $\cF_{i_1}, \dots, \cF_{i_j}$, of $\cF$ having $\bbF_{q^m}$ as their best friend. Then it suffices to see that, if $\beta^l\notin\stabbeta(\cF)=\langle\beta\rangle\cap \bbF_{q^m}^\ast$, then  $\beta^l$ does not stabilize the subspaces $\cF_{i_1}, \dots, \cF_{i_j}$. Consequently, we have
		$$
		d_f(\cF, \cF\beta^l) \geq \sum_{k=1}^j d_S(\cF_{i_k}, \cF_{i_k}\beta^l)\geq 2mj.
		$$	
		To prove $(2)$, let us start assuming that there are at least two different subspaces $\cF_i$ and $\cF_j$ in $\cF$ with $\bbF_{q^m}$ as their best friend. By $(1)$ we have that $d_f(\orbbeta(\cF))\geq 4m> 2m$. On the other hand, suppose that no subspace in $\cF$ has $\bbF_{q^m}$ as its best friend. In this case, for every $1\leq i\leq r$, we put $\bbF_{q^{m_i}}$ the best friend of $\cF_i$ and, since $\bbF_{q^m}=\bigcap_{i=1}^r\bbF_{q^{m_i}}$, we have that $m$ is a proper divisor of every $m_i$. In particular, $m < m_i$, for every $1\leq i\leq r$. Now, for every $\beta^l\notin\stabbeta(\cF)$, we have at least one index $1\leq i\leq r$ such that $\beta^l\notin\stabbeta(\cF_i)$. Hence,
		$$
		d_f(\cF, \cF\beta^l) \geq d_S(\cF_i, \cF_i\beta^l)\geq 2m_i > 2m.
		$$
		Thus, $d_f(\cC)> 2m$.

	\end{proof}

		\begin{remark} Note that the converses of statements  $(2)$ and $(1)$ in the previous result  are not necessarily true. Take, for instance, $\cF=(\bbF_{q^2}, \bbF_{q^4}, \bbF_{q^8})$ on $\bbF_{q^{16}},$ which has best friend $\bbF_{q^2}$. Let us consider $\beta=\alpha^5$ where $\langle \alpha \rangle = \bbF_{q^{16}}^\ast$. Then we have just one subspace of $\cF$ with best friend $\bbF_{q^2}$ whereas $d_f(\orbbeta(\cF))=12=2 \cdot 3 \cdot 2  \geq  4$. At the same time there are not three subspaces in $\cF$ sharing its best friend.
			
	\end{remark}
	
	The previous theorem allows us to discard some type vectors if we work with the minimum value of the  distance when the best friend is $\bbF_{q^m}$.
	
	\begin{corollary}
		Let $\cF$ be a flag of type $(ms_1, \dots, ms_r)$ on $\bbF_{q^n}$ with best friend $\bbF_{q^m}$ and take $\beta\in\bbF_{q^n}^\ast\setminus \bbF_{q^m}^\ast$. If $d_f(\orbbeta(\cF))=2m$, then $\gcd(s_j, \frac{n}{m})\neq 1$ for, at least $r-1$ indices $1\leq j\leq r$.
	\end{corollary}
	\begin{proof}
		By means of Theorem \ref{th: dist min type vector}, we know that there is exactly one subspace of $\cF$ having best friend $\bbF_{q^m}$, say $\cF_i$. Now, for each $j\neq i$, we put $\bbF_{q^{m_j}}$ the best friend of $\cF_j$. In particular, we know that $m$ is a proper divisor of every $m_j$. Let us write $m_j= m a_j$, with $a_j>1$ for every $j\neq i$.  In addition, $m_j=ma_j$ divides both $\dim(\cF_j)=ms_j$ and $n$. Hence, $1 < a_j$ divides both $s_j$  and $\frac{n}{m}$. We conclude that $\gcd(s_j, \frac{n}{m})> 1$ for all $1\leq j\leq r$, $j\neq i$.
	\end{proof}	
	
	\begin{example}
		 If $n=16$ and we fix $\bbF_{q^2}$ as the best friend of our flags, the minimum distance value $4$ cannot be obtained for type $(4,6,10)$ since $\gcd(3, 8)=\gcd(5, 8)=1$. In contrast, this value would be attainable for type $(4,6,8)$, for instance.
			Using the same argument, if we take $n=14$, and consider a flag $\cF$ on $\bbF_{q^{14}}$ having the subfield $\bbF_{q^2}$ as its best friend, we can conclude that $\beta$-cyclic orbit flag codes generated by $\cF$ will never give distance $4$, unless $\cF$ is the flag of length one $\cF=(\bbF_{q^2})$. 
	\end{example}

	We have seen that, fixed the best friend $\bbF_{q^m}$, the minimum value of the distance $2m$ can only be obtained by codes $\orbbeta(\cF)$ in which $\cF$ has exactly one subspace with $\bbF_{q^m}$ as its best friend as well. On the other end, as said in Theorem \ref{theo: caracterización ODFC}, a flag code $\cC$ attains the maximum possible distance for its type if, and only if, it is disjoint and all its projected codes attain the respective maximum (subspace) distance. Recall that a flag code $\cC$ of length $r$ on $\bbF_{q^n}$ is disjoint if it holds
	\begin{equation}\label{eq:disjunto}
		|\cC_1|=\dots=|\cC_r|=|\cC|.  
	\end{equation}
	In \cite[Prop. 4.19]{FlagCyclic} the authors prove that for cyclic orbit flag codes ($\beta$ primitive) this condition is equivalent to say that each subspace of $\cF$ has the same best friend (then the best friend of $\cF$). Summing up, we can also draw conditions on the type vector in the case of cyclic orbit flag codes having $\bbF_{q^m}$ as their best friend and the largest possible distance, that is, the upper value of the range in (\ref{eq: distances interval}).
	
	\begin{proposition}
		\cite[Cor. 4.23]{FlagCyclic}\label{prop:type vector optimum distance}
		Assume that the cyclic orbit code $\orb(\cF)$ is an optimum distance flag code on $\bbF_{q^n}$ with the subfield $\bbF_{q^m}$ as its best friend. Then one of the following statements holds:
		\begin{enumerate}
			\item $\orb(\cF)$ is a constant dimension code of dimension either $m$ or $n-m$.
			\item $\orb(\cF)$ has type vector $(m, n-m)$.
		\end{enumerate}
		In any of the three cases above, the code $\orb(\cF)$ has the largest possible size, that is, $\frac{q^n-1}{q^m-1}$.
	\end{proposition}
	
	Using Theorem \ref{th: dist min type vector}, we obtain the next construction of cyclic orbit flag codes with the best possible distance for the above mentioned cases.
	
	\begin{proposition}\label{prop: construction odfc}
		Let $\cF=(\cF_1, \cF_2)$ a flag of type $(n, n-m)$ on $\bbF_{q^n}$. If $\cF_1$ and $\cF_2$ have the subfield $\bbF_{q^m}$ as their best friend, then the cyclic orbit codes $\orb(\cF_1)$, $\orb(\cF_2)$ and $\orb(\cF)$ have the maximum possible distance.
	\end{proposition}
	\begin{proof}
		The result holds for $\orb(\cF_1)=\orb(\bbF_{q^m})$ by means of (\ref{def: spread Anna-Lena}). For $\orb(\cF_2)$, it suffices to see that, if $\bbF_q^m$ is the best friend of $\cF_2$, then $d_S(\orb(\cF_2))= 2m$, which is the maximum possible distance for dimension $n-m$. Last, by means of Theorem \ref{th: dist min type vector}, we conclude that $d_f(\orb(\cF))=4m$, i.e., the maximum possible distance for type $(m, n-m)$. 
	\end{proof}	
	\begin{remark}
		In the case of $\beta$-cyclic orbit flag codes with $\beta$ non primitive, in \cite[Prop. 4.19]{FlagCyclic} it is proved that the code $\orbbeta(\cF)$ is disjoint if, and only if,
			
			\begin{equation}\label{eq:disjoint beta}
				\langle\beta\rangle\cap\bbF_{q^m}^\ast= \langle\beta\rangle\cap\bbF_{q^{m_1}}^\ast= \dots = \langle\beta\rangle\cap\bbF_{q^{m_r}}^\ast.
			\end{equation}
			Note that to have (\ref{eq:disjoint beta}) it is not necessary that all the subspaces of $\cF$ share the same best friend, contrary to what happens if $\beta$ primitive (see part $(2)$ on Example \ref{ex: flags with subfields}). Moreover, in \cite[Thm. 4.21]{FlagCyclic} the authors give also conditions on the type vector of $\cF$ of an optimum distance $\beta$-cyclic flag code with fixed best friend if $\beta$ is not primitive.
			Here we present some examples extracted from \cite[Table 3]{FlagCyclic} where they 
			determine the set of allowed dimensions in the type vector, depending on the size of the acting subgroup $\langle \beta\rangle$ of $\bbF_{2^{12}}^\ast=\langle\alpha\rangle$, when the best friend is $\bbF_{2^{2}}$.
		\begin{table}[H] 
			\centering
			\begin{small}
				\begin{tabular}{cccccc}
					\hline
					$\beta$           & $|\beta|$ & $\langle\beta\rangle\cap\bbF_{q^m}^\ast$ & $|\orbbeta(\cF)|$ & Allowed dimensions  & Max. distance    \\ \hline
					$\alpha$               & 4095   &     $\bbF_{2^2}^\ast$		      &    1365     & 2, 10             & 8 \\
					$\alpha^5$             & 819    &     $\bbF_{2^2}^\ast$           &    273      & 2, 4, 8, 10       & 24 \\
					$\alpha^{9}$           & 455    &     $\{1\}$                     &    455      & 2, 10             & 8 \\
					$\alpha^{63}$          & 65     &     $\{1\}$                     &    65       & 2, 4, 6, 8, 10    & 36 \\
					\end{tabular}
			\end{small}
			\caption{Admissible dimensions for $q=2$, $n=12$, $m=2$.}\label{table: q=2, n=12, m=2}
		\end{table}
		Concerning the explicit construction of such codes, in \cite{OrbitODFC,MA-X}, the authors follow the approach of \cite{TrautManBraunRos2013} to build optimum distance flag codes under the action of (subgroups of) Singer groups of the special linear group and the general linear group, respectively, by placing a suitable spread among the projected codes. In our framework, this idea corresponds to the choice a generating flag that has certain subfield among its subspaces. Let us exhibit some concrete examples. 
		\begin{example}\label{ex: flags with subfields} Let us work in $\bbF_{2^{12}}$ and fix $\bbF_{q^2}$ as best friend of all our flags.
			\begin{enumerate}
				\item Take $\cF=(\bbF_{2^2}, \cF_2, \bbF_{2^6}, \cF_4, \cF_5)$ of type $(2,4,6,8,10)$ and consider $\beta= \alpha^{63}$, then the orbit $\orbbeta(\cF)$ is an optimum distance flag code of cardinality $65$ (see \cite{OrbitODFC}), which is the maximum possible size for an optimum distance flag code of this type. With the same notation, the orbit $\orbbeta((\bbF_{2^2}, \cF_2,\cF_4, \cF_5))$ is an optimum distance flag code, in this case of type $(2,4,8,10)$, of the same size.
				\item On the other hand, following the ideas in \cite{MA-X}, if we consider the flag $\cF'=(\bbF_{2^2}, \bbF_{2^4}, \cF'_3, \cF'_4)$ of type $(2,4,8,10)$ such that $\bbF_{2^4}$ is the best friend of $\cF'_3$, and take $\beta=\alpha^5$, then the orbit $\orbbeta(\cF')$ is an optimum distance flag code with cardinality $273$. Note that in this example the subspaces $\cF'_1$ and $\cF'_3$ do not share their best friend even thought $\orbbeta(\cF')$ is disjoint.
			\end{enumerate}
		\end{example}
	\end{remark}	
	
	These examples lead us to study  $\beta$-cyclic orbit codes when we place one or more subfields in the generating flag.

	\subsection{Generating flags based on subfields}
	In this subsection we focus on $\beta$-cyclic orbit flag codes  generated by flags having at least one subfield among their subspaces. We distinguish two situations: either every subspace in the generating flag is a subfield or there is also one subspace that is not a subfield.
	
	\subsubsection{Galois flag codes}
	
	Let us start with $\beta$-cyclic orbit flag codes generated by flags having just subfields of $\bbF_{q^n}$ as subspaces, that is, generated by the so-called \emph{Galois flags}.
	This particular class of $\beta$-cyclic orbit flag codes was introduced in \cite{FlagCyclic}. Let us recall the definition. Consider a sequence of integers $1\leq t_1 < \dots < t_r < n$  such that all of them are divisors of $n$ and  $t_i$ divides $t_{i+1}$, for $1\leq i \leq r-1$.

	\begin{definition} 
		The \emph{Galois flag} of type $(t_1, \dots, t_r)$ on $\bbF_{q^n}$ is the flag given by the sequence of nested subfields $(\bbF_{q^{t_1}}, \dots, \bbF_{q^{t_r}})$. Given $\beta \in \bbF_{q^n}^*$, the $\beta$-cyclic orbit flag code generated by this Galois flag is called the $\beta$-\textit{Galois cyclic orbit flag code}, or just $\beta$-\textit{Galois flag code}, for short, of type $(t_1, \dots, t_r)$.
	\end{definition}

	In the Galois flag $\cF$ of type vector $(t_1, \dots, t_r)$, clearly the $i$-th subspace has the subfield $\bbF_{q^{t_i}}$ as best friend. Hence, the first subfield $\bbF_{q^{t_1}}$ is the best friend of any  $\beta$-Galois flag code of type $(t_1, \dots, t_r)$. For $\beta$ primitive we have the following straightforward result.
	
	\begin{proposition}[\cite{FlagCyclic}]\label{prop:Size Galois codes}
		Let $\cC$ be the Galois flag code of type $(t_1, \dots, t_r)$, then the cardinality of this flag code is $|\cC|=(q^n-1)/(q^{t_1}-1)$ and its distance is $d_f(\cC)=2t_1$. Its $i$-th projected code $\cC_i$ has size $|\cC_i|=(q^n-1)/(q^{t_i}-1)$ and distance $2t_i$.
	\end{proposition}
	
	\begin{remark}\label{rem:Relacion subgrupo distancia}
		Note that, if we take the Galois flag  $\cF$ of type $(t_1, \dots, t_r)$, the distance $d_f(\orb(\cF))=2t_1$ is the lowest possible one for cyclic orbit flag codes with $\bbF_{q^{t_1}}$ as best friend, according to (\ref{eq: dist max flags}) (in case $|\cC|>1$). In fact, there is a precise set of attainable distances for the different orbits $\orbbeta(\cF)$ when we consider the action of subgroups $\langle \beta \rangle \subseteq \bbF_{q^n}^*$. Furthermore, as proved in \cite{FlagCyclic}, we can always select  $\beta$  in a controlled manner such that the code $\orbbeta(\cF)$ reaches any distance value in the set of possible distances. This choice is made by checking the relationship between the subgroup $\langle \beta \rangle$ and the subfields $\bbF_{q^{t_i}}$. Let us recall the precise result.
	\end{remark}
	\begin{theorem}\cite[Thm. 4.14]{FlagCyclic} \label{theo: distance Galois beta cyclic}
		Let $\cF$ be the Galois flag of type $(t_1, \dots, t_r)$ and consider an element $\beta\in \bbF_{q^n}^\ast$. Then 
		\begin{equation}\label{eq:distances Galois codes}
			d_f(\orbbeta(\cF)) \in \{ 0, 2t_1, 2(t_1+t_2),\dots, 2(t_1+t_2+\dots+t_r)\}.
		\end{equation}
		Moreover,
		\begin{enumerate}
			\item $d_f(\orbbeta(\cF))= 0$ if, and only if, $\stabbeta(\cF_1)=\stabbeta(\cF_r)=\langle\beta\rangle$. \label{theo: distance Galois beta cyclic item 1}
			\item $d_f(\orbbeta(\cF))= 2\sum_{i=1}^r t_i$ if, and only if, $\stabbeta(\cF_1)=\stabbeta(\cF_r)\neq\langle\beta\rangle$. \label{theo: distance Galois beta cyclic item 2}
			\item $d_f(\orbbeta(\cF))= 2\sum_{i=1}^{j-1} t_i$  if, and only if, $\stabbeta(\cF_1)\neq\stabbeta(\cF_r)$ and $j\in\{2, \dots, r\}$ is the minimum index such that $\stabbeta(\cF_1) \subsetneq \stabbeta(\cF_j).$ \label{theo: distance Galois beta cyclic item 3}
		\end{enumerate}
	\end{theorem}
	
	In view of the previous result, it is worth highlighting that, given a Galois flag $\cF$, the range of attainable distances by the codes $\orbbeta(\cF)$ follows a concrete pattern in terms of the dimensions in the type vector (see (\ref{eq:distances Galois codes})). On the other hand, we have the possibility to gradually improve the distance of $\orb(\cF)$ by selecting the orbit $\orbbeta(\cF)$ for an appropriate $\beta$, even if this choice could involve a loss of size. This nice behaviour gives rise to think that Galois codes could constitute an appropriate ``skeleton'' to support a more general family of $\beta$-cyclic orbit flag codes whose properties, in turn, might be driven by them. To explore this idea, in the following section we introduce a new family of codes.

	\subsubsection{Generalized Galois flag codes}\label{subsec: generalized Galois}
	
	 Let us take now generating flags having at least one subspace that is a subfield and at least another one that is not. This condition gives length at least two. Note also that, all the fields in a flag $\cF$ constitute a Galois subflag.
		
		\begin{definition}\label{def: genelized Galois}
			We say that a flag $\cF=(\cF_1, \dots, \cF_k)$ of type $(s_1, \dots, s_k)$ \emph{generalizes the Galois flag of type $(t_1, \dots, t_r)$} if $\{ t_1, \dots, t_r\} \subsetneq \{ s_1, \dots, s_k\}$ and the following conditions are satisfied:
			\begin{enumerate}
				\item The subflag of $\cF$ of type $(t_1, \dots, t_r)$ is  the Galois flag of this type,
				\item there is at least one subspace of $\cF$ that is not a field.
			\end{enumerate}
	\end{definition}
	
	\begin{remark}
		Observe that the second condition in Definition \ref{def: genelized Galois} excludes Galois flags from our study of generalized Galois flags. Even more, according to the previous definition, a generalized Galois flag is just a flag having at least one field and one subspace that is not a field among its subspaces. Besides, in the conditions of the previous definition, $\cF$ clearly generalizes every subflag of the Galois flag of type $(t_1, \dots, t_r)$ as well. We pay special attention to the longest Galois flag being a subflag of $\cF$. 
	\end{remark}
	
	\begin{definition}
		Let $\cF$ be a generalized Galois flag. Its longest Galois subflag is called its \em{underlying Galois subflag}.
	\end{definition}

	Observe that the underlying Galois subflag of a generalized Galois flag always exists and, due to the nested structure of flags, it is unique. 
	
	\begin{definition}
		Given $\cF$ a generalized Galois flag and $\beta \in \bbF_{q^n}^*$, the $\beta$-cyclic orbit flag code generated by $\cF$ is called a \emph{generalized $\beta$-Galois (cyclic orbit) flag code}.  If $\cF'$ is the underlying Galois subflag of $\cF$, then we say that $\orbbeta(\cF')$ is the \emph{underlying $\beta$-Galois flag code} of $\orbbeta(\cF)$.
	\end{definition}
	
	Let us see some examples reflecting different situations related to this new class of flag codes.
	\begin{example}
		Take $n=8$ and primitive elements $\alpha\in\bbF_{q^8}$, $\gamma\in\bbF_{q^4}$. The sequences
		$$
		\cF= (\bbF_{q^2}, \bbF_{q^4}, \bbF_{q^4} \oplus \bbF_q^2 \alpha) \ \ \text{and} \ \ \cF'= (\bbF_{q^2}, \bbF_{q^2} \oplus \bbF_q \gamma, \bbF_{q^4})
		$$
		are generalized flags of type $(2,4, 6)$  and $(2,3,4)$ on $\bbF_{q^8}$, respectively, with common underlying Galois subflag $(\bbF_{q^2}, \bbF_{q^4})$. 
		Now, for any $\beta \in \bbF_{q^n}^*$, the best friend of the $\beta$-Galois flag code $\orbbeta((\bbF_{q^2}, \bbF_{q^4}))$ is the field $\bbF_{q^2}$ and this property also holds for the generalized $\beta$-Galois flag code $\orb(\cF)$. However, the best friend of  $\orb(\cF')$ is $\bbF_q$, which is a field not appearing in $\cF'$. 
	\end{example}

	\begin{remark}
		Given a generalized $\beta$-Galois flag code $\cC$ of type $(s_1, \dots, s_k)$ with underlying Galois subflag $(\bbF_{q^{t_1}}, \dots, \bbF_{q^{t_r}})$, if the subfield $\bbF_{q^m}$ is the best friend of $\cC$, then it holds $\bbF_{q^m} \subseteq \bbF_{{q}^{t_1}}$. 
		
		Concerning the attainable distance values for this class of codes, contrary what happens with Galois flag codes whose set of reachable distances is completely determined by the type, different situations can arise. 
	\end{remark}
	
	\begin{example}\label{ex: 2 examples distances}
		Take $n=4$ and consider a generalized Galois flag $\cF$ of type $(1,2,3)$ with underlying Galois subflag $(\bbF_q)$. In this case, the set of attainable values for $d_f(\orb_\beta(\cF))$ is exactly the same as for general flag codes of this type on $\bbF_{q^4}$, that is,  any even integer between $0$ and $8$. However, for the same choice of the parameters, flags generalizing the Galois flag $(\bbF_q, \bbF_{q^2}),$ i.e., those of the form 
		$$
		\cF' = (\bbF_q, \bbF_{q^2}, \cF'_3),
		$$
		for some subspace $\cF'_3$ of dimension $3$ of $\bbF_{q^4}$, present a restriction on the set of possible distances. Let us prove that the value $d_f(\orb_\beta(\cF'))=6$ cannot be obtained for any $\beta\in\bbF_{q^n}^\ast$. Observe that the projected code $\orb_\beta(\bbF_{q^2})$ is a partial spread of dimension $2$ of $\bbF_{q^4}$. Hence, when computing the distance
		$$
		d_f(\cF', \cF'\beta^i) = d_S(\bbF_q, \bbF_q\beta^i) +  d_S(\bbF_{q^2}, \bbF_{q^2}\beta^i) +  d_S(\cF'_3, \cF_3'\beta^i),
		$$ 
		we have that 
		$$
		d_S(\bbF_{q^2}, \bbF_{q^2}\beta^i) = \left\lbrace
		\begin{array}{cl}
			0 & \text{if} \ \beta\in\bbF_{q^2}^\ast, \\
			4 & \text{otherwise.}
		\end{array}
		\right.
		$$
		Moreover, if $d_S(\bbF_{q^2}, \bbF_{q^2}\beta^i)=0$, then it holds $d_f(\cF', \cF'\beta^i)\leq 4$. On the other hand, if $d_S(\bbF_{q^2}, \bbF_{q^2}\beta^i)=4,$ we have $\dim( \bbF_{q^2} \cap \bbF_{q^2}\beta^i)=0$ or, equivalently, $\dim( \bbF_{q^2} \oplus\bbF_{q^2}\beta^i)=4$. Hence, for this precise $\beta^i$, we get
		$$
		\dim(\bbF_q \cap \bbF_q \beta^i)=0 \ \ \text{and} \ \ \dim(\cF'_3 + \cF'_3 \beta^i)=4.
		$$
		In both cases, we can conclude $d_S(\bbF_q, \bbF_q\beta^i)=d_S(\cF'_3, \cF'_3\beta^i)=2$. As a consequence, it holds $d_f(\cF', \cF'\beta^i)=8$.
		Thus, the value $d=6$ cannot be obtained if we consider flags of type $(1,2,3)$ on $\bbF_{q^4}$ generalizing $(\bbF_{q}, \bbF_{q^2})$.
	\end{example}
	
	The situation exhibited in the last example is a direct consequence of the presence of certain subfields of $\bbF_{q^n}$ as subspaces of a generalized Galois flag $\cF$. In other words, its underlying Galois subflag affects, in some sense, the value $d_f(\orbbeta(\cF))$.  The next result establish some conditions on the minimum distance of generalized $\beta$-Galois flag codes that allow us to discard some values of the distance.
	
	\begin{theorem}\label{theo: potential values dist}
		Let $\cF$ be a generalized Galois flag  of type $(s_1, \dots, s_k)$ on $\bbF_{q^n}$ with underlying Galois subflag $(\bbF_{q^{t_1}}, \dots, \bbF_{q^{t_r}})$ and take $\beta\in \bbF_{q^n}^\ast$.  Consider $i \in \{1, \ldots,r\}$. Then the following statements hold:
		\begin{enumerate}
			\item If $\beta\in\bbF_{q^{t_i}}^\ast$, then $d_S(\cF_l, \cF_l\beta)= 0,$ for all $s_l\in\{t_i, \dots, t_r\}$.
			\item If $\beta\notin\bbF_{q^{t_i}}^\ast,$ then $d_S(\cF_l, \cF_l\beta)= 2s_l$ for all $s_1\leq s_l \leq t_i$.
		\end{enumerate}
	\end{theorem}
	\begin{proof} 
		\begin{enumerate} 
			\item Assume that $\beta\in\bbF_{q^{t_i}}^\ast$. Hence, $\beta\in\bbF_{q^{t_j}}^\ast,$ for every $i\leq j\leq r$. In other words, we have $\bbF_{q^{t_j}}\beta=\bbF_{q^{t_j}}$, i.e., $d_S(\bbF_{q^{t_j}}, \bbF_{q^{t_j}}\beta)=0$ for all $i\leq j\leq r$.
			\item Take now $\beta\notin\bbF_{q^{t_i}}^\ast=\stab(\bbF_{q^{t_i}})$. In this case, $\bbF_{q^{t_i}}$ and $\bbF_{q^{t_i}}\beta$ are different subspaces in the $t_i$-spread $\orb(\bbF_{q^{t_i}})$  and we have that $\dim(\bbF_{q^{t_i}}\cap \bbF_{q^{t_i}}\beta)=0$. Thus, for every dimension $s_l\leq t_i$ in the type vector, we have $\dim(\cF_l\cap\cF_l\beta)=0$ and then $d_S(\cF_l, \cF_l\beta)=2s_l$.
		\end{enumerate}
	\end{proof}
	
	According to this result, it is clear that some combinations of subspace distances are automatically discarded when we compute the minimum distance of a generalized $\beta$-Galois flag code. Even thought we do not have a pattern to compute the distance values as it occurs for Galois flag codes (see (\ref{eq:distances Galois codes}), by means of Theorem \ref{theo: potential values dist}, we can state some required conditions for potential distance values between flags on a generalized $\beta$-Galois flag code.
	
	\begin{definition}\label{def: potential distances}
		Take a flag $\cF$ of type $(s_1, \dots, s_k)$ on $\bbF_{q^n}$ generalizing the Galois flag $(\bbF_{q^{t_1}}, \dots, \bbF_{q^{t_r}})$ and an element $\beta\in\bbF_{q^n}^\ast$. We say that an even integer $d$ is a \emph{potential value for $d_f(\orbbeta(\cF))$} if it can be obtained as a sum of subspace distances of dimensions $s_1, \dots, s_k$ satisfying:
		\begin{enumerate}
			\item For dimensions $t_i$, only distances $0$ or $2t_i$ are considered.
			\item If we sum $2t_i$ for dimension $t_i$, then all the distances for lower dimensions in the type vector are maximum as well.
			\item If for some dimension $t_i$ we have distance $0$, then the same happens for dimensions $t_j$, with $i\leq j\leq r$.
		\end{enumerate}
	\end{definition}
	
	\begin{remark}
		Notice that, according to Definition \ref{def: potential distances}, and as suggested in Example \ref{ex: 2 examples distances}, having the field $\bbF_q$ as the first subspace of a flag does not affect to the set of potential distance values since every $\beta$-cyclic code of dimension $1$ of $\bbF_{q^n}$ (generated or not by $\bbF_q$) has distance either $0$ or $2$.  On the other hand, as also mentioned in Example \ref{ex: 2 examples distances}, some distances cannot be attained when we have other fields among the subspaces of the generating flag. For instance, the single value $d=6$ is discarded for $n=4$, type $(1,2,3)$ and underlying Galois subflag of type $(1,2)$.
	\end{remark}
	The next example shows that, in general, many values of the flag distance are not compatible with the underlying structure of nested fields.
	
	\begin{example}
		Fix $n=16$ and the type vector $(2,4,5,6,8)$. In general, every even integer $0\leq d \leq 50$ is a possible value for the flag distance for this choice of the parameters. Nevertheless, if $\cF$ is a generalized Galois flag of type $(2,4,5,6,8)$ with underlying Galois subflag of type $(2,4,8)$, then for every $\beta\in\bbF_{q^n}^\ast$, the set of potential values for $\orbbeta(\cF)$ is
		$$
		\{0,2,4,6,8,10,22,50\}.
		$$
		In other words, no intermediate distances $12\leq d\leq 20$ or $24\leq d\leq 48$ can be obtained when the starting flag contains fields as its subspaces of dimensions $2, 4$ and $8$ (for more details on the computation of the set of distance values, we refer the reader to \cite{cotas}, where a deep study on the flag distance parameter and its behaviour is presented).
	\end{example}
	
	\noindent At this point, we can assert that the distance of a generalized Galois flag code is strongly influenced by its underlying Galois code. Hence, it is quite natural to wonder if generalized $\beta$-Galois flag codes behave as well as Galois flag codes in the following sense: 
	\begin{quote}
		$(\ast)$ \emph{Given a generalized Galois flag $\cF$ and a potential value $d$ for the distance defined in \ref{def: potential distances}, can we always find a suitable subgroup $\langle \beta \rangle \subseteq \bbF_{q^n}^*$ such that $d_f(\orbbeta(\cF))=d$? }
	\end{quote}
	In the following section we prove that the answer to the previous question is negative by exhibiting a specific family of generalized Galois flags.
	
	\section{A construction of generalized Galois  flag codes}\label{sec: A construction}
	
	This section is devoted to build generalized Galois flags written in a regular form that allows us to provide constructions of $\beta$-cyclic orbit flag codes with a prescribed best friend. In Subsections \ref{subsec: basic}, \ref{subsec:weaved} and \ref{subsec: decoding}, we focus on the particular case of $\beta$ primitive due to the fact that the obtained cyclic orbit codes present important properties that deserve to be underlined. Finally, Subsection \ref{subsec:weaved beta} is devoted to deal with the $\beta$-cyclic case and, in particular, to give an answer to the question $(\ast)$ formulated in the previous section.
	
	Recall that, according to the definition of best friend, we can express all the subspaces of a given flag as vector spaces over its best friend (see (\ref{eq: sum direct F_i})). In our case, we will consider a specific family of flags, whose subspaces are written in the regular form used in Proposition \ref{prop: suma directa}. Let us describe the form of such flags and obtain the parameters of cyclic orbit flag codes generated by them.
	
	Fix $\bbF_{q^m}$ a subfield of $\bbF_{q^n}$ and consider a primitive element $\alpha$ of $\bbF_{q^n}.$ For each positive integer $l$ such that $1\leq l < \frac{q^n-1}{q^m-1},$ let $L$ be the degree of the minimal polynomial of $\alpha^l$ over $\bbF_{q^m}$. Observe that $L$ is also the degree of the field extension $\bbF_{q^{mL}}/\bbF_{q^m}$, that is, $L=[\bbF_{q^{mL}}:\bbF_{q^m}].$ Hence, $L$ divides $[\bbF_{q^n}:\bbF_{q^m}]=n/m=s$ and, we have that $L\leq s$. In addition, the set $\{1, \alpha^l, \alpha^{2l}, \ldots, \alpha^{(L-1)l}\}$ is a basis of the field extension $\bbF_{q^{mL}}/ \bbF_{q^m}.$ Thus, we can write  
	
	\begin{equation}\label{eq: last subspace is a field}
		\bigoplus_{j=0}^{L-1}\bbF_{q^m}\alpha^{jl} = \bbF_{q^m}[\alpha^l] \cong \bbF_{q^{mL}}.
	\end{equation}
	Now, for every  $i=1, \ldots,  L $, the vector space 
	\begin{equation}\label{eq: especial sum direct F_i tipo completo}
		\cU_i = \bigoplus_{j=0}^{i-1} \bbF_{q^m}\alpha^{jl}
	\end{equation}
	has dimension $mi$ (over $\bbF_{q}$) and, as stated in Section \ref{sec:Preliminaries}, it is a field if, and only if, either $i=1$ or $i=L$. Hence, the sequence $(\cU_1, \dots, \cU_L)$ forms a generalized Galois flag of type $(m, 2m, \ldots, mL)$ with underlying Galois subflag $(\bbF_{q^m}, \bbF_{q^{mL}})$.
	
	\begin{remark}\label{rem: cota r}
		Observe that we cannot define a direct sum of this shape with more than $L$ terms, since every power of $\alpha^l$ is always an element in $\cU_L=\bbF_{q^m}[\alpha^l] \cong \bbF_{q^{mL}}$. Hence, this regular form allows us to construct flag codes on $\bbF_{q^n}$  of length $r\leq L$. Moreover, in case $L=s$, as $\cU_L=\bbF_{q^n}$, we just get $r\leq L-1=s-1$.
	\end{remark}

	The following two subsections (\ref{subsec: basic} and \ref{subsec:weaved}) are devoted to describe a construction of generalized Galois flag codes having $\bbF_{q^m}$ as their best friend. We perform it in two steps. First, we use flags in the regular form just described above in order to obtain a ``basic'' construction of generalized Galois flag codes where the underlying Galois flag code has, at most, length $2$. Then, we propose a procedure to overcome this restriction and present another construction of generalized Galois flag code having a prescribed underlying Galois flag code by suitably ``weaving'' several basic generalized Galois flag codes.  
	
	\subsection{Basic constructions}\label{subsec: basic}
	By means of Proposition \ref{prop: BF en forma regular}, we can easily determine the best friend of the subspaces $\cU_i$ defined in (\ref{eq: especial sum direct F_i tipo completo}): it is the subfield $\bbF_{q^m}$, for $1\leq i\leq L-1$ whereas the subspace $\cU_L=\bbF_{q^{mL}}$ is its own best friend. This fact implies that we will find one or two fields among the subspaces $\cU_i$, according to Remark \ref{rem: cota r}.
	Since the case of length $r=1$ corresponds to constant dimension codes (already studied in \cite{GLMoTro2015}), from now on, we will assume $r\geq 2$. Now, we know that for every type vector given by multiples of $m$, say $(ms_1, \dots, ms_r)$, where $1\leq s_1 < \dots < s_r \leq L$, we select the subspaces defined in (\ref{eq: especial sum direct F_i tipo completo}) corresponding respectively to these dimensions, that is,
	\begin{equation}\label{eq: subespacios del flag teorema clasificación}
		\cF_i = \cU_{s_i} = \bigoplus_{j=0}^{s_i-1} \bbF_{q^m}\alpha^{lj}, \ 1\leq i\leq r. 
	\end{equation}
	With this notation, the next result holds. 
	
	\begin{theorem}\label{theo: teorema clasificación}
		Let $\alpha$ be a primitive element of $\bbF_{q^n}$, $l$ a positive integer with $1\leq l < \frac{q^n-1}{q^m-1}$ and $L$ the degree of the minimal polynomial of $\alpha^l$ over $\bbF_{q^m}$. Consider the flag $\cF=(\cF_{1}, \dots, \cF_{r})$  of type $(ms_1, \dots, ms_r)$ on $\bbF_{q^n}$ with subspaces defined in (\ref{eq: subespacios del flag teorema clasificación}). Hence, the code $\orb(\cF)$ has best friend $\bbF_{q^m}$ and, in particular, its cardinality  is $(q^n-1)/(q^m-1).$  Moreover, 
		\begin{enumerate}
			\item If $s_r < L$, the code $\orb(\cF)$ is consistent with distance $d_f(\orb(\cF))=2mr$.  \label{theo: teorema clasificación item1}
			\item If $s_r = L$, we have that $d_f(\orb(\cF))=2m(r-1)$ and we can write 
			$$
			\orb(\cF)= \dot\bigcup_{i=0}^{c-1} \mathrm{Orb}_{\alpha^c}(\cF\alpha^i), 
			$$
			with $c=\frac{q^n-1}{q^{mL}-1}$. \label{theo: teorema clasificación item2}
		\end{enumerate}
	\end{theorem}
	\begin{proof}
		Recall that, by means of Proposition \ref{prop: BF en forma regular}, the best friend of every subspace in $\cF$ is either  $\bbF_{q^m}$ or $\bbF_{q^{mL}}$. Since $r\geq 2$, there is at least one subspace with $\bbF_{q^m}$ as best friend and, automatically, this subfield is the best friend of the flag $\cF$. As a consequence, by Proposition \ref{prop: cardinality and best friend}, the cardinality of $\orb(\cF)$ is $(q^n-1)/(q^m-1)$.
		
		Now, suppose that $s_r<L$. In this case, every dimension in the type vector is $ms_i\leq ms_r < mL$ and every subspace in the flag $\cF$ has the subfield $\bbF_{q^m}$ as its best friend. Hence, the code is disjoint and, by means of Theorem \ref{th: dist min type vector}, we have $d_f(\orb(\cF))\geq 2mr$.	Moreover, notice that, for every $1\leq i\leq r$, the subspace $\cF_i\cap\cF_i\alpha^l=\bigoplus_{j=1}^{s_i-1} \bbF_{q^m}\alpha^{lj}$ has dimension $m(s_i-1)$ over $\bbF_q$ and then $d_S(\orb(\cF_i))=d_S(\cF_i, \cF_i\alpha^l) = 2m.$
		Hence, it holds
		$$
		d_f(\cC)\leq d_f(\cF, \cF\alpha^l) = \sum_{i=1}^r d_S(\cF_i, \cF_i\alpha^l) = 2mr
		$$
		and we conclude $d_f(\cC)=2mr.$ Moreover, the value coincides with the sum of the ones of its $r$ projected codes (each one of them with distance $2m$). Hence, our code is  consistent (see Definition \ref{def: consistentes}).
		
		Assume now that $s_r=L$. Then $\orb(\cF)$ is not disjoint since the subspace $\cF_r=\bbF_{q^{mL}}$ is its own best friend. However, since $\cF_1, \dots, \cF_{r-1}$ have the same best friend $\bbF_{q^m}$,  Theorem \ref{th: dist min type vector} ensures that $d_f(\orb(\cF))\geq 2m(r-1)$. On the other hand, observe that $\alpha^l\in \bbF_q^{mL}\setminus \bbF_{q^m}^\ast$ stabilizes $\cF_r$ but $d_S(\cF_i, \cF_i\alpha^l)=2m$ for every $1\leq i <r$. Hence, we have $d_f(\orb(\cF))=d_f(\cF, \cF\alpha^l)=2m(r-1)$.
		
		Concerning the structure of this cyclic orbit flag code, it is clear that 
		$$
		\bigcup_{i=0}^{c-1} \mathrm{Orb}_{\alpha^c}(\cF\alpha^i) \subseteq \orb(\cF).
		$$
		Let us see that both sets have the same cardinality. To do so, observe that $\alpha^c$ is a primitive element of $\bbF_{q^{mL}}$. Hence, for every $1\leq i\leq r$, it holds
		$$
		\mathrm{Stab}_{\alpha^c}(\cF \alpha^i) = \bbF_{q^{mL}}^\ast \cap \stab(\cF \alpha^i)= \bbF_{q^{mL}}^\ast \cap \bbF_{q^m}^\ast= \bbF_{q^m}^\ast.
		$$
		As a consequence, we have $|\mathrm{Orb}_{\alpha^c}(\cF\alpha^i)|=\frac{q^{mL}-1}{q^m-1}$. Now, we prove that all these orbits are different. To do so, for every choice $0\leq i \leq c-1$, observe that flags in $\mathrm{Orb}_{\alpha^c}(\cF\alpha^i)$ have the same last subspace $\cF_r\alpha^i=\bbF_{q^{mL}}\alpha^i$ since $\alpha^c$ stabilizes $\cF_r=\bbF_{q^{mL}}=\{0\}\cup\langle \alpha^c\rangle$. Moreover, the last projected code 
		$$
		\orb(\cF_r)=\{ \bbF_{q^{mL}}\alpha^i \ | \ 0 \leq i\leq c-1 \}
		$$
		of $\orb(\cF)$ is precisely the $mL$-spread $\orb(\bbF_{q^{mL}})$ of $\bbF_{q^n}$. Hence, for every choice $0\leq i < j \leq c-1$, subspaces $\bbF_{q^{mL}}\alpha^i$ and $\bbF_{q^{mL}}\alpha^j$ are different. Therefore, all the orbits $\mathrm{Orb}_{\alpha^c}(\cF\alpha^i)$ are different and the cardinality of their union is exactly
		$$
		c \cdot \frac{q^{mL}-1}{q^m-1} = \frac{q^n-1}{q^{mL}-1} \cdot \frac{q^{mL}-1}{q^m-1} = \frac{q^n-1}{q^m-1} = |\cC|.
		$$
	\end{proof}
	
	From Theorem \ref{theo: teorema clasificación}, and making a suitable choice of the type vector, we derive some constructions of our interest.
	
\begin{corollary}\label{cor: clasifiación Galois y ODFC}
		Consider the flag $\cF$ of type $(ms_1, \dots, ms_r)$ with $r\geq 2$ defined in (\ref{eq: subespacios del flag teorema clasificación}).  If $L < s$, then:
		\begin{enumerate}
			\item the code $\orb(\cF)$ is a Galois flag code if, and only if, the type vector is $(m, mL)$.
			\item $\orb(\cF)$ is a generalized Galois flag code if, and only if, $$\emptyset\neq\{1, L\}\cap\{s_1, \dots, s_r\}\neq\{1, L\}.$$
		\end{enumerate}
		\end{corollary}
		\begin{proof}
		In the first place, if $L<s$, it is clear that $\cF$ is a Galois flag if it just have subfields of $\bbF_{q^n}$ as its subspaces. According to expression  (\ref{eq: subespacios del flag teorema clasificación}), just subspaces of dimensions $m$ and $mL$ are fields. Moreover, since $r\geq 2,$ the result follows. 
	\end{proof}
		
\begin{corollary}\label{cor: clasifiación Galois y ODFC 2}
Take the flag $\cF$ of type $(ms_1, \dots, ms_r)$  defined in (\ref{eq: subespacios del flag teorema clasificación}) and assume $r\geq 2$ and $L=s$, then:
		\begin{enumerate}
			\item $\orb(\cF)$ is a generalized Galois flag code if, and only if, $s_1=1$.
            \item In particular, if $s_1=1$ and $s_2=L-1$, then $\orb(\cF)$ is an optimum distance generalized Galois flag code of type $(m, n-m)$ with the largest possible size.
		\end{enumerate}
	\end{corollary}
	\begin{proof}
		In this case, the only subfield of $\bbF_{q^n}$ writen in the regular form (\ref{eq: especial sum direct F_i tipo completo}) is $\bbF_{q^m}$ and the first statement follows straightforwardly. For the second one, it suffices to notice that the subspaces $\cF_1=\cU_1=\bbF_{q^m}$ and $\cF_2=\cU_{L-1}$ are of dimensions $m$ and $n-m$ and have $\bbF_{q^m}$ as their best friend. Hence, the result holds by means of Proposition \ref{prop: construction odfc}.
	\end{proof}
	
	At this point, we have all the ingredients to perfectly describe the structure of any cyclic orbit flag code $\orb(\cF)$ given in Theorem \ref{theo: teorema clasificación}, in case its distance is either the minimum or the maximum possible one. This result is closely related with to discussion in Subsection \ref{subsec:interdependence} about the interdependence of distance values and type vectors, for this particular family of codes.
	\begin{theorem}\label{prop:caracterizacion maxima y minima distancia}  
		Let $\alpha$ be a primitive element of $\bbF_{q^n}$, $l$ a positive integer such that $1\leq l < \frac{q^n-1}{q^m-1}$, and $L$ the degree of the minimal polynomial of $\alpha^l$ over $\bbF_{q^m}$. Consider $\cF=(\cF_1, \ldots, \cF_r)$ the flag of type $(ms_1, ms_2, \dots, ms_r)$ on $\bbF_{q^n}$ of length $r\geq 2$ with subspaces defined as in (\ref{eq: subespacios del flag teorema clasificación}). If $s=n/m$, then:
		\begin{enumerate}
			\item The code $\orb(\cF)$ has distance equal to $2m$ if, and only if, its type vector is $(ms_1, mL)$ for some $1\leq s_1 < L < s $. Moreover, if $s_1=1$, then $\orb(\cF)$ is the Galois flag code of type $(m, mL)$.\label{prop: caracterización item 1}
			\item The code $\orb(\cF)$ is an optimum distance flag code if, and only if, $L=s$ and its type vector is $(m, m(L-1)).$ In this case, $\orb(\cF)$ is a generalized Galois flag code that attains the largest possible size. \label{prop: caracterización item 2}
		\end{enumerate}
	\end{theorem}
	
	\begin{proof}
		We divide the proof into two parts.
		\begin{enumerate}

			\item  By means of Theorem \ref{th: dist min type vector}, if the orbit $\orb(\cF)$ has distance $2m$, then there is exactly one subspace of $\cF$ with $\bbF_{q^m}$ as its best friend. Since subspaces defined in (\ref{eq: subespacios del flag teorema clasificación}) have the subfield $\bbF_{q^m}$ as their best friend except if they are fields, we conclude that the last subspace of the generating flag must be the field $\bbF_{q^{mL}}=\bbF_{q^m}[\alpha^l]$. Thus $L$ cannot be $s$ and it holds $L<s$. Its first subspace can be any other subspace $\cF_1$ of dimension $ms_1< mL$.
			
			To prove the converse, just note that the distance between $(\cF_1, \bbF_{q^{mL}})$ and $(\cF_1, \bbF_{q^{mL}})\alpha^l$ is $2m$, which is the minimum distance for cyclic orbit flag codes with $\bbF_{q^m}$ as their best friend.
			
			To finish, if $s_1=1$, the only possibility for $\cF$ is to be the Galois flag of type $(m, mL)$ and the result holds.
			
			\item On the other hand, if $\orb(\cF)$ is an optimum distance cyclic orbit flag code with $\bbF_{q^m}$ as its best friend, by means of Proposition \ref{prop:type vector optimum distance}, and assuming $r\geq 2$, its type vector must be $(m, n-m)$. Hence, we need $mL$ to be at least $n-m=m(s-1)$. In other words, $L$ must be greater or equal than $s-1$. However, $L$ has to divide $s$. If $L=s-1$, the only possibility is $s=2$ and the only type vector consisting of multiples of $m$ is $(m)$, which has length one. Thus, it must hold $L=s$. The converse is also true by application of Corollary \ref{cor: clasifiación Galois y ODFC}.
		\end{enumerate}
		
	\end{proof}
	\begin{remark}
		Concerning the first statement of the previous result, it is important to point out that if we consider a generating flag not necessarily written in the regular form described in (\ref{eq: subespacios del flag teorema clasificación}), it is possible to attain the distance $2m$ with no other field among the subspaces of the generating flag than the best friend of the flag. Recall that, for general flags, even if we fix the field $\bbF_{q^m}$ as the first subspace, there are three possibilities for the best friend of each one of its subspaces:  the field $\bbF_{q^m}$, the subspace itself (in case it is a field) or an intermediate extension field over $\bbF_{q^m}$. The next example contemplates this situation.
		
		\begin{example}
			Fix $n=16$ and consider a generalized Galois flag $\cF=(\cF_1, \cF_2)$ with type vector $(2,8)$ such that
			$$
			\begin{array}{cclll}
				\cF_1 & = & \bbF_{q^2}, & & \\ 
				\cF_2 & = & \bbF_{q^2} \oplus \bbF_{q^2}\beta \oplus \bbF_{q^2}\alpha \oplus  \bbF_{q^2}\beta\alpha & = & \bbF_{q^4}\oplus \bbF_{q^4}\alpha,
			\end{array}
			$$
			where $\alpha$ denotes a primitive element of $\bbF_{q^{16}}$ and $\beta=\alpha^{(q^{16}-1)/(q^4-1)}$ is a primitive element of the subfield $\bbF_{q^4}$. Observe that $\bbF_{q^4}$ is a friend of $\cF_2$. Even more, it is its best friend by means of Proposition \ref{prop: BF en forma regular}, since the degree of the minimal polynomial of $\alpha$  over $\bbF_{q^4}$ is $16/4=4$. According to this, $\bbF_{q^2}$ is the best friend of the cyclic orbit flag code $\orb(\cF)$.
			Moreover, we have that $d_f(\cF, \cF\beta)=d_S(\cF_1, \cF_1\beta)= 4$ is the minimum possible distance for cyclic orbit flag codes with $\bbF_{q^2}$ as its best friend. However, $\orb(\cF)$ is not the Galois flag code of type $(2,8)$, since $\cF$ is the only flag in the code with its first subspace containing the element $1\in\bbF_{q^{16}}$ but $\cF_2$ is not a field.
		\end{example}
	\end{remark}
	
	Observe now that, taking again the flag $\cF=(\cF_1, \cF_2)$ as in the previous example, we have the subfield $\bbF_{q^4}$ as an intermediate subspace. In other words, the sequence 
	$$
	\cF'=(\bbF_{q^2}, \bbF_{q^4}, \bbF_{q^4}\oplus\bbF_{q^4}\alpha)
	$$
	forms a generalized Galois flag longer than $\cF$. Notice that the field $\bbF_{q^4}$ can be written in regular form as a vector space over $\bbF_{q^2}$ as $\bbF_{q^2}\oplus\bbF_{q^2}\beta$. Moreover, the subspace $\bbF_{q^4}\oplus\bbF_{q^4}\alpha$ is, at the same time, written in regular form as a vector space over $\bbF_{q^4}$. Inspired by this idea, we will describe a general procedure that allows us to obtain generalized Galois flags codes, written in regular form over suitably chained subfields.

	\subsection{Weaving basic generalized Galois flag codes}\label{subsec:weaved}
	The previous basic construction (Theorem \ref{theo: teorema clasificación}) presents a limitation on the number of subfields that can appear as subspaces of the generating flag. In this subsection, we focus on a systematic construction of generalized Galois flag codes with a prescribed underlying Galois subflag. More precisely, if $m_1, m_2, \ldots, m_k$ are divisors of $n$ such that $m_i$ divides $m_{i+1}$, for every $1\leq i \leq k$, we work on the construction of generalized Galois flag codes with $(\bbF_{q^{m_1}}, \dots, \bbF_{q^{m_k}})$ as underlying Galois subflag. 
	
	As a matter of notation, through this section we will write $m_{k+1}=n$. Let $\alpha$ be a primitive element of $\bbF_{q^n}$ and put $c_i=\frac{q^n-1}{q^{m_i}-1}$, for $1\leq i\leq k+1$. It turns out that each power $\alpha_i=\alpha^{c_i}$ is a primitive element of the corresponding subfield $\bbF_{q^{m_i}}.$ For every $2 \leq i\leq k+1$, the degree of the minimal polynomial of $\alpha_i$ over $\bbF_{q^{m_{i-1}}}$ is $L_i=\frac{m_i}{m_{i-1}}$. With this notation, we consider $k$ flags $\cF^1, \dots, \cF^k$ on $\bbF_{q^n},$ whose subspaces are given by
	\begin{equation}\label{eq: flag construccion general}
		\cF_j^i= \bigoplus_{l=0}^{j-1} \bbF_{q^{m_i}}\alpha_{i+1}^l, 
	\end{equation}
	for $1\leq j\leq L_{i+1}-1$ and $1\leq i \leq k$.
	Observe that, for every $1\leq i\leq k$, we have that $\cF_1^i= \bbF_{q^{m_i}}$ and the dimension of $\cF_j^i$ (as an $\bbF_q$-vector space) is $jm_i$. Hence, the type vector of $\cF^i$ is given by all the multiples of $m_i$ smaller than $m_{i+1}$, that is, $(m_i, 2m_i, \dots, m_{i+1}-m_i)$. As a consequence of Theorem \ref{theo: teorema clasificación}, the next result holds.
	
	\begin{corollary}\label{cor: construccion general 1}
		Let $\cF^i$ be the flag defined in (\ref{eq: flag construccion general}) for every $1\leq i\leq k$. Then the generalized Galois flag code $\orb(\cF^i)$ is consistent with distance $2(m_{i+1}-m_i)$ and has the field $\bbF_{q^{m_i}}$ as its best friend.
	\end{corollary}
	\begin{proof}
		For every $1\leq i\leq k$, we apply Theorem \ref{theo: teorema clasificación} (part (\ref{theo: teorema clasificación item1})) to the generalized Galois flag code $\orb(\cF^i)$ and conclude that it is a consistent flag code with distance equal to $2m_i(L_{i+1}-1)=2(m_{i+1}-m_i)$. 
	\end{proof}
	
	Observe that we have constructed a collection of orbit flag codes with respective generating flags in regular form over their best friend. The first subspace of each flag $\cF^i$ is precisely its best friend and contains the last subspace of the previous flag, since 
	$$
	\cF_{L_{i+1}-1}^i \subseteq \bbF_{q^{m_i}}[\alpha_{i+1}] = \bbF_{q^{m_{i+1}}} = \cF_1^{i+1},
	$$ for every value $1\leq i\leq k-1$. By means of this property, we can consider a generating generalized Galois flag having all the subfields $\{\bbF_{q^{m_i}}\}_{i=1}^{k}$ among its subspaces just by taking
	\begin{equation}\label{eq: flag construccion general pegado}
		\cF  = (\cF_1^1, \dots, \cF_{{L_1}-1}^1, \cF_1^2, \dots, \cF_{{L_2}-1}^2, \dots, \cF_1^k, \dots, \cF_{{L_k}-1}^k),
	\end{equation}
	whose type vector is $(m_1, \dots, m_2-m_1, m_2, \dots, m_3-m_2, m_3, \ldots, m_k, \ldots, n-m_k).$ In this way, by weaving the independent basic constructions described in Corollary \ref{cor: construccion general 1}, we get a generalized Galois flag code with the prescribed tower of subfields $(\bbF_{q^{m_1}}, \dots, \bbF_{q^{m_k}})$ as its underlying Galois subflag.

	\begin{proposition}\label{prop: weaved construction}
		Let $\cF$ be the generalized Galois flag on $\bbF_{q^n}$ given in (\ref{eq: flag construccion general pegado}). Then the generalized Galois flag code $\orb(\cF)$  generalizes the Galois flag of type $(m_1, \dots, m_r)$. Its cardinality is $\frac{q^n-1}{q^{m_1}-1}$ and its minimum distance, $2(m_2-m_1)$. 
	\end{proposition}
	\begin{proof}
		By construction, it is clear that the subspace of dimension $m_i$ of $\cF$ is the field $\bbF_{q^{m_i}}$, for every  $1\leq i\leq k$. Moreover, subspaces of dimensions $m_i, \ldots, m_{i+1}-m_{i}$ have the subfield $\bbF_{q^{m_i}}$ as its best friend. As a result, the best friend of the flag $\cF$ coincides with its first subspace, that is, $\bbF_{q^{m_1}}$. This fact leads to the statement about the cardinality. Let us compute now the minimum distance of the code. First, by means of Theorem \ref{th: dist min type vector}, since there are exactly $L_2-1$ subspaces of $\cF$ with $\bbF_{q^{m_1}}$ as their best friend, we conclude that $d_s(\orb(\cF))\geq 2m_1(L_2-1)=2(m_2-m_1)$. Moreover, observe that $\alpha_2$ stabilizes every subspace of the flag $\cF$ containing $\bbF_{q^{m_2}}=\{0\}\cup\langle \alpha_2\rangle$. Hence, 
		$d_f(\orb(\cF))\leq d_f(\cF, \cF\alpha_2)=d_f(\cF^1, \cF^1\alpha_2)= 2m(L_2-1) = 2(m_2-m_1)$
		since, for every $1\leq j \leq L_2-1$, the subspace $\cF_j^1\cap\cF_j^1\alpha= \bigoplus_{l=1}^{j-1} \bbF_{q^{m}}\alpha_{2}^l$ has dimension $m(j-1).$ We conclude that $d_f(\orb(\cF))=2(m_2-m_1)$.
	\end{proof}
	
	\begin{remark}
		Note that weaving our basic constructions allows us to give generalized Galois flag codes with any given underlying Galois subflag in a systematic way. Another interesting fact to point out is that the best friends of the subspaces in the generalized Galois flag $\cF$ defined in (\ref{eq: flag construccion general pegado}) form a nested sequence of subfields. This does not happen in general for arbitrary generalized Galois flag codes and it helps us to easily determine the cardinality and distance of the code $\orb(\cF)$ as well as to give bounds for the distance when we consider its $\beta$-cyclic subcodes, as we will see in Subsection \ref{subsec:weaved beta}. On the negative side, contrary to what happens with the basic construction, the waved one is not consistent since it is not even disjoint.
	\end{remark}
	
	\subsection{Decoding our constructions over the erasure channel}\label{subsec: decoding}
	
	The use of flags in network coding was originally introduced by Liebhold \textit{et al.} in \cite{LiebNebeVaz18}. In that paper, a channel model for flags was presented and some constructions, together with their decoding algorithms (over the erasure channel) were provided. In \cite{Consistentes}, a decoding algorithm over the erasure channel for consistent flag codes is presented. In particular, such an algorithm can be applied to the basic construction given in Theorem \ref{theo: teorema clasificación}, part (\ref{theo: teorema clasificación item1}). Although the rest of constructions in this paper are not consistent, we can adapt the decoding process in \cite{Consistentes} to them. To do so, let us briefly recall some concepts related to the notion of \emph{correctability}.
	
	Assume that we have sent a flag $\cF=(\cF_1, \dots, \cF_r)$  and hence, the receiver gets a sequence of nested subspaces $\cX=(\cX_1, \dots, \cX_r)$ that, when working over an erasure channel, must satisfy $\cX_i\subseteq\cF_i$, for all $1\leq i\leq r$. In this context, each value $e_i=d_S(\cF_i, \cX_i)=\dim(\cF_i)-\dim(\cX_i)$ is called \emph{number of erasures at the $i$-th shot}  whereas $e=d_f(\cF, \cX)=\sum_{i=1}^r e_i$ is the \emph{total number of erasures}. We say that the total number of erasures $e$ is \emph{correctable} (by minimum distance) by a flag code $\cC$ whenever $e\leq \lfloor \frac{d_f(\cC)-1}{2}\rfloor$. Analogously, we also say that the value $e_i$ is correctable by the projected code $\cC_i$ if $e_i\leq \lfloor \frac{d_S(\cC_i)-1}{2}\rfloor.$
	
	Let us fix the flag code $\cC$ as the one presented in Theorem \ref{theo: teorema clasificación}, part (\ref{theo: teorema clasificación item2}). Recall that such a code has distance $2m(r-1)$ and $r$ projected codes of distance $2m$. Following the ideas of \cite[Proposition 8]{Consistentes}, we state the next result.
	\begin{proposition}\label{prop: correctability basic}
		If the total number of erasures $e$ is correctable by the generalized Galois flag code $\cC$, then there exists some $1\leq i\leq r-1$ such that the value $e_i$ is also correctable by the corresponding projected code $\cC_i$.
	\end{proposition}
	\begin{proof}
		Assume that no value $e_i$ is correctable for every $1\leq i\leq r-1$. Equivalently, we have that $e_i\geq m$ for every $1\leq i\leq r-1$. As a consequence, we have that
		$$
		e=\sum_{i=1}^r e_i \geq m(r-1)+e_r  \geq m(r-1),
		$$
		which is a contradiction, since $\cC$ can correct up to $m(r-1)-1$ erasures.
	\end{proof}
	
	Now, if $\cC'$ denotes the generalized Galois flag code obtained by the weaved construction given in Proposition \ref{prop: weaved construction} and $m_1=m$, then  we have $d_f(\cC')=2(m_2-m)=2m(L_2-1)$. Moreover, $d_S(\cC'_i)=2m$ holds for the first $L_2-1$ projected codes. Hence, the same argument in the proof of Proposition \ref{prop: correctability basic} can be used to show that a correctable total number of erasures can be detected and corrected by one of the first $L_2-1$ projected codes of $\cC$.  
	\begin{proposition}\label{prop: correctability weaved}
		If the total number of erasures $e$ is correctable by $\cC'$, then there exists some $1\leq i\leq L_2-1$ such that $e_i$ is correctable by the projected code $\cC'_i$.
	\end{proposition}
	Moreover, in both situations, every projected code has the same distance, which is $2m$. Hence, the number of erasures at any shot is correctable whenever it holds $e_i\leq m-1$. We can easily identify if an erasure is correctable just by checking the dimension of every received subspace $\cX_i$. The next proposition is valid for both a generalized Galois flag codes $\cC$ and $\cC'$.
	\begin{proposition}
		The number of erasures $e_i$ is correctable by the constant dimension code $\cC_i$ (resp. $\cC'_i$) if, and only if, $\dim(\cX_i) \geq \dim(\cF_i)-m+1$.
	\end{proposition}
	\begin{proof}
		Assume that we send a flag $\cF\in\cC$ (resp. in $\cC'$) and a stuttering flag $\cX$ is received. Then $e_i$ is correctable by $\cC_i$ (resp. $\cC'_i$) if, and only if, it holds 
		$$
		\dim(\cF_i)-\dim(\cX_i)=d_S(\cF_i, \cX_i)= e_i \leq m-1
		$$ 
		or equivalently, if $\dim(\cX_i) \geq \dim(\cF_i)-m+1$.
	\end{proof}
	
	\begin{remark}
		Observe that neither $\cC$ nor $\cC'$ are disjoint flag codes. However, Proposition \ref{prop: correctability basic} (resp. \ref{prop: correctability weaved}) allows us to decode at least one of the received subspaces $\cX_i$ into the sent one $\cF_i$ for an index $i$ satisfying $|\cC_i|=|\cC|$ (resp. $|\cC'_i|=|\cC'|$). Hence, after having recovered $\cF_i$, one can easily obtain the sent flag $\cF$ as the unique flag in $\cC$ (resp. $\cC'$) having $\cF_i$ as its $i$-th subspace.
	\end{remark}
	
	\subsection{The $\beta$-cyclic case}\label{subsec:weaved beta}
	In this part of the paper, we consider orbits under the action of proper subgroups of $\bbF_{q^n}^\ast$ generated by the flag $\cF$ given in (\ref{eq: flag construccion general pegado}), which has underlying Galois subflag $(\bbF_{q^{m_1}}, \bbF_{q^{m_2}}, \dots, \bbF_{q^{m_k}})$. In other words, for every $\beta\in\bbF_{q^n}^\ast$, we study the generalized $\beta$-Galois flag code $\orbbeta(\cF)$. Recall that this code has type $(m_1, \dots, m_2-m_1, m_2, \dots, m_3-m_2, m_3, \ldots, m_k, \ldots, n-m_k)$ and it has the following particularity: the best friends of the subspaces of $\cF$ are nested. More precisely, the subfield $\bbF_{q^{m_i}}$ is the best friend of the subspaces of dimensions $m_i,\dots, m_{i+1}-m_i$ in the flag, for every $1\leq i\leq k$, where $m_{k+1}=n$. This property makes our flag $\cF$ be closer to the Galois flag of type $(m_1, \dots, m_k)$ than other flags that also generalize it. As a result, we can give lower and upper bounds for the distance of $\orbbeta(\cF)$ by studying the sequence of subgroups
	$$
	\langle \beta\rangle \cap\bbF_{q^{m_1}}^\ast \subseteq \langle \beta\rangle \cap \bbF_{q^{m_2}}^\ast \subseteq \dots \subseteq \langle \beta\rangle \cap \bbF_{q^{m_k}}^\ast.
	$$
	In particular, we consider two possibilities: either all these subgroups coincide or some inclusion is strict. In the latest case, we are specially interested in the first index $1< i\leq k$ such that $\langle \beta\rangle \cap \bbF_{q^{m_1}}^\ast\neq \langle \beta\rangle \cap \bbF_{q^{m_i}}^\ast$. Moreover, we exclude those elements $\beta\in\bbF_{q^{m_1}}^\ast$ since they provide trivial orbit flag codes with distance equal to zero.
	
	\begin{theorem}\label{theo: beta cyclic weaved construction}
		Let $\cF$ be the generalized Galois flag given in (\ref{eq: flag construccion general pegado}) and $\beta\in\bbF_{q^n}^\ast\setminus\bbF_{q^{m_1}}^\ast$. For every $1\leq i\leq k$, we write $M_i=\sum_{j=1}^{i-1} m_{j+1}(L_{j+1}-1)$.
		\begin{enumerate}
			\item If $\langle \beta\rangle \cap \bbF_{q^{m_1}}^\ast = \langle \beta\rangle \cap \bbF_{q^{m_k}}^\ast$, then 
			\begin{equation*} 
				2m_k(L_{k+1}-1)+M_k \leq d_f(\orbbeta(\cF))
				\leq m_k\left\lfloor \frac{L_{k+1}^2}{2}\right\rfloor+M_k.
			\end{equation*}  \label{theo: beta cyclic weaved construction item 1}
			\item Otherwise, consider the minimum $1< i\leq k$ such that $\langle \beta\rangle \cap \bbF_{q^{m_1}}^\ast \subsetneq \langle \beta\rangle \cap \bbF_{q^{m_i}}^\ast$. Then it holds:
			\begin{equation*} 
				2m_{i-1}(L_{i}-1)+M_{i-1} \leq d_f(\orbbeta(\cF)) 
				\leq  m_{i-1}\left\lfloor \frac{L_{i}^2}{2}\right\rfloor+M_{i-1}.
			\end{equation*} 
			 \label{theo: beta cyclic weaved construction item 2}
		\end{enumerate}
	\end{theorem}
	\begin{proof}
		Assume that $\langle \beta\rangle \cap \bbF_{q^{m_1}}^\ast = \langle \beta\rangle \cap \bbF_{q^{m_k}}^\ast$. Let us compute the distance $d_f(\cF, \cF\beta^l)$, for every element $\beta^l\notin\stabbeta(\cF)=\langle \beta\rangle \cap \bbF_{q^{m_1}}^\ast$. 
		Since $\langle \beta\rangle \cap \bbF_{q^{m_1}}^\ast = \langle \beta\rangle \cap \bbF_{q^{m_k}}^\ast$, this power $\beta^l$ does not stabilize any subspace in the flag $\cF$. In particular, observe that $\bbF_{q^{m_k}}$ and $\bbF_{q^{m_k}}\beta^l$ are different subspaces in the spread $\orb(\bbF_{q^{m_k}})$. In other words, it holds $d_S(\bbF_{q^{m_k}},\bbF_{q^{m_k}}\beta^l)=2m_{k}$ and, by means of Theorem \ref{theo: potential values dist}, every subspace distance between subspaces of $\cF$ and $\cF\beta^l$ of dimensions lower than $m_k$ is maximum as well, i.e., twice the corresponding dimension. Hence, for dimensions up to $m_k-m_{k-1} =m_{k-1}(L_k-1),$ we obtain the sum of subspace distances:
		$$
		\begin{array}{ccl}
			\sum_{j=1}^{k-1}
			\big( 2m_j + \dots + 2m_j(L_{j+1}-1) \big)  
			\Big)
			&=&
			\sum_{j=1}^{k-1}
			2m_j\big( 1 + \dots +(L_{j+1}-1)\big) \\  
			&=&
			\sum_{j=1}^{k-1}
			m_jL_{j+1}(L_{j+1}-1)  \\ 
			&=&
			\sum_{j=1}^{k-1}
			m_{j+1}(L_{j+1}-1)=M_k. 
		\end{array}
		$$
		Moreover, since subspaces of dimensions $m_k, \dots, n-m_k=m_k(L_{k+1}-1)$ have $\bbF_{q^{m_k}}$ as their best friend, if $d$ represents the flag distance between the subflags of type $m_k(1, \dots, L_{k+1}-1)$ of $\cF$ and $\cF\beta^l$ and, by means of (\ref{eq: distance bounds}), we have: 
		$$
		2m_k(L_{k+1}-1) \leq d \leq m_k \left\lfloor \frac{L_{k+1}^2}{2}\right\rfloor.
		$$
		Combining these two facts, we get the desired lower and upper bounds for $d_f(\cF, \cF\beta^l)$,  if $\beta^l\notin\stabbeta(\cF)$. In particular, these bounds are also valid for $d_f(\orbbeta(\cF))$.
		
		To prove (\ref{theo: beta cyclic weaved construction item 2}), suppose that $\langle\beta\rangle\cap\bbF_{q^{m_1}}\neq \langle\beta\rangle\cap\bbF_{q^{m_k}}$ and then take the minimum $1< i \leq k$ such that $\langle\beta\rangle\cap\bbF_{q^{m_1}}\subsetneq \langle\beta\rangle\cap\bbF_{q^{m_i}}$. 
		In this case, we can always find an element $\beta^l\in \bbF_{q^{m_i}}^\ast\setminus \bbF_{q^{m_1}}^\ast$. This power $\beta^l$ stabilizes every subspace in $\cF$ having the subfield $\bbF_{q^{m_i}}$ as a friend, i.e., all those of dimensions at least $m_i$. This means that these dimensions do not contribute to the computation of $d_f(\cF, \cF\beta^l)$. On the other hand, since $\langle\beta\rangle\cap\bbF_{q^{m_1}}= \dots = \langle\beta\rangle\cap\bbF_{q^{m_{i-1}}}$, then $\bbF_{q^{m_{i-1}}}\neq \bbF_{q^{m_{i-1}}}\beta^l$ are different spread elements and the distance between them is $2m_{i-1}$. As before, by means of Theorem \ref{theo: potential values dist}, all the subspace distances are maximum for dimensions up to $m_i$. In particular, the distance between the subflags of type $(m_1, \dots, m_i-m_{i-1})$ of $\cF$ and $\cF\beta^l$ is exactly
		$$
		\sum_{j=1}^{i-2} m_{j+1}(L_{j+1}-1)=M_{i-1}.
		$$
		Besides, observe that the subspaces of dimensions $m_{i-1}, \dots, m_i-m_{i-1}$ of $\cF$ and $\cF\beta^l$ are $\bbF_{q^{m_{i-1}}}$-subspaces of $\bbF_{q^{m_i}}$. Hence, if $d$ denotes the distance between the corresponding subflags of $\cF$ and $\cF\beta^l$, by (\ref{eq: distance bounds}), we have
		
		$$
		2m_{i-1}(L_{i}-1) \leq d \leq m_{i-1}\left\lfloor \frac{L_i^2}{2}\right\rfloor.
		$$
		As a result, we conclude
		\begin{equation}
			\label{eq: bound distance beta-cyclic}
			2m_{i-1}(L_{i}-1)+M_{i-1} \leq d_f(\cF, \cF\beta^l)
			\leq   m_{i-1} 
			\left\lfloor\frac{L_i^2}{2}\right\rfloor+M_{i-1}
		\end{equation}
		for every $\beta^l\in\bbF_{q^{m_i}}^\ast\setminus\bbF_{q^{m_{1}}}^\ast$. Arguing as above, if we take another power of $\beta$ not in $\bbF_{q^{m_i}}^\ast$, say $\beta^h$, we obtain maximum subspace distances up to, at least, dimensions $m_i$ and then 
		$$
		d_f(\cF, \cF\beta^h)\geq d_f(\cF, \cF\beta^l).
		$$
		As a consequence, the minimum distance of the code is attained when we consider powers  $\beta^l \in\bbF_{q^{m_i}}^\ast\setminus\bbF_{q^{m_{1}}}^\ast$ and we have the result.
	\end{proof}	
	
	Observe that the upper bound for the distance in the first part of Theorem \ref{theo: beta cyclic weaved construction} is exactly the maximum possible distance for general flags of the corresponding type. Here below, we go a step further and give a sufficient condition for our construction to provide optimum distance flag codes, i.e., flag codes with the maximum possible distance for their type on $\bbF_{q^n}$.
	
	\begin{corollary}\label{cor: odfc beta-cyclic}
		Consider the generalized Galois flag  $\cF$ given in (\ref{eq: flag construccion general pegado}) and take $\beta\in\bbF_{q^n}^\ast$ such that $\langle \beta\rangle \cap \bbF_{q^{m_1}}^\ast = \langle \beta\rangle \cap \bbF_{q^{m_k}}^\ast$. If $L_{k+1}\leq 3$, then  $\orbbeta(\cF)$ is an optimum distance flag code.
	\end{corollary}
	\begin{proof}
		Notice that, under these assumptions, by means of Theorem  \ref{theo: beta cyclic weaved construction} (part (\ref{theo: beta cyclic weaved construction item 1})), we have
	 
		\begin{equation*} 
			2m_k(L_{k+1}-1)+M_k \leq d_f(\orbbeta(\cF)) 
			\leq   m_k\left\lfloor \frac{L_{k+1}^2}{2}\right\rfloor +M_k.
		\end{equation*}
		Moreover, if the degree of the extension $\bbF_{q^n}/\bbF_{q^{m_k}}$, that is, the positive integer $L_{k+1}=[\bbF_{q^n}:\bbF_{q^{m_k}}]= n/m_{k+1}$, satisfies $1< L_{k+1} \leq 3$, then we have 
		$$
		2(L_{k+1}-1)=\left\lfloor \frac{L_{k+1}^2}{2}\right\rfloor.
		$$
		
		Hence, both lower and upper bounds for $d_f(\orbbeta(\cF))$ coincide and the code $\orbbeta(\cF)$ attains the maximum possible distance for its type vector.
	\end{proof}
	
	To finish this subsection we address the question $(\ast)$ launched in Subsection \ref{subsec: generalized Galois}. Recall that the potential distance values of generalized Galois flag code follow the rules stated in Theorem \ref{theo: potential values dist} and Definition \ref{def: potential distances}. These conditions arise naturally from the presence of certain subfields among the subspaces of a generalized Galois flag. Concerning question $(\ast)$, we wonder if, given a generalized Galois flag $\cF$, every potential value of the distance can be truly obtained by a cyclic (or $\beta$-cyclic) orbit flag code generated by $\cF$. We answer this question by using the $\beta$-cyclic construction presented in Subsection \ref{subsec:weaved beta}
	
	\begin{example}
		Consider the following parameters choice: $q=2$, $n=10$. Moreover, we take nested subfields $\bbF_2 \subset \bbF_{2^5}$ of the field $\bbF_{2^{10}}$, which correspond to the election of divisors $m_1=1$ and $m_2=5$ of $n=10$. In this case, we have $L_2=5$ and $L_3=2$. Let us use the generalized Galois flag $\cF=(\bbF_2,\cF_2, \cF_3, \cF_4, \bbF_{2^5})$ of type $(1,2,3,4,5)$ with subspaces:
		$$
		\cF_2=\bbF_{2}\oplus \bbF_{2}\gamma, \qquad \cF_3=\bbF_{2}\oplus \bbF_{2}\gamma\oplus \bbF_{2}\gamma^2, \quad \text{and} \quad \cF_4=\bbF_{2}\oplus \bbF_{2}\gamma\oplus \bbF_{2}\gamma^2\oplus  \bbF_{2}\gamma^3,
		$$
		where $\gamma$ is a primitive element of $\bbF_{2^5}$.
		
		The set of potential values of the distance in this case is given by:
		$$
		\{ 0, 8, 10, 12, 30 \}.
		$$
		In this case, we know how to choose $\beta\in\bbF_{2^{10}}^\ast$ so that the orbit $\orbbeta(\cF)$ attains some of these distances. More precisely:
		\begin{itemize}
			\item Distance $d=0$ is obtained if, and only if, $\beta\in\bbF_{2}^\ast=\{1\}=\stab(\cF)$.
			\item For distance $d=8$, it suffices to take the cyclic orbit code $\orb(\cF)$ that, by means of Theorem \ref{theo: teorema clasificación}, has distance $d_f(\orb(\cF))=8$.
			\item Last, since $L_3=2\leq 3$, by application of Corollary \ref{cor: odfc beta-cyclic}, we know that every $\beta\in\bbF_{2^{10}}^\ast$ such that $\langle\beta\rangle\cap\bbF_{2^5}=\{1\}$ makes $\orbbeta(\cF)$ be an optimum distance flag code. For instance, it suffices to consider subgroups $\langle\beta\rangle$ of $\bbF_{2^{10}}^\ast$ of orders $\{3,11,33\}$ to attain the maximum distance, i.e., the value $d=30$.
		\end{itemize}
		Moreover, for this specific example, we have obtained the parameters of the code $\orbbeta(\cF),$ for every subgroup $\langle\beta\rangle$ of $\bbF_{2^{10}}^\ast$ by using GAP. First of all, since $\stab(\cF)=\bbF_2^\ast=\{1\}$, we have $|\orbbeta(\cF)|=|\beta|$. The next table collects the set of distances for the generating flag $\cF$:

		\begin{table}[H]
			\begin{center}
				\begin{tabular}{cc}
					\hline
					$|\beta|$ & $d_f(\orbbeta(\cF))$ \\ \hline
					1         & 0                   \\
					3         & 30                  \\
					11        & 30                  \\
					31        & 8                   \\
					33        & 30                  \\
					93        & 8                   \\
					341       & 8                   \\
					1023      & 8    \\ \hline                
				\end{tabular}
				\caption{Distance of all the $\beta$-cyclic orbit flag code generated by $\cF$.}\label{tab: table distances}
			\end{center}
		\end{table}

		Using this example, one can see that not all the potential values of the distance can be obtained by taking a suitable subgroup of $\bbF_{2^{10}}^ \ast$. 	It suffices to observe that neither distances $d=10$ nor $d=12$ appear in Table \ref{tab: table distances}. Even more, despite the fact that $d_f(\orbbeta(\cF)) \neq 12$ for any $\beta \in \bbF_{2^{10}}^\ast$, this value still can be the distance between a couple of flags; for instance, we have
		$$
		d_f(\cF, \cF\gamma^2) = 2+4+4+2+0= 12.
		$$
However, this is not even true for distance $d=10$. In other words, for every $\beta\in\bbF_{2^{10}}^\ast$ and any power $1\leq l\leq |\beta|$, the distance $d_f(\cF, \cF\beta^l) \neq 10$.
	\end{example}
	
	\section{Conclusions and future work}
	
	In this work we present new contributions to the study of $\beta$-cyclic orbit flag codes started in \cite{FlagCyclic}, also following the viewpoint of \cite{GLMoTro2015}. The best friend of a flag code still has a crucial role throughout the paper. In particular, we discuss the rich interplay among flag distances, best friend and type vector for this family of codes.
	
	Nevertheless, whereas in \cite{FlagCyclic} the accent was put precisely on the best friend of the flag code, this time we turn our attention to the generating flag of the orbit. We focus specially on those ones having at least one field among their subspaces, by distinguishing the case of having just fields on the generating flag from the case where also at least one subspace not being a field appears. This dichotomy leads, on one side, to the known $\beta$-Galois flag codes and, on the other one, to the generalized $\beta$-Galois flag codes, which properties we describe.
	
	Every generalized $\beta$-Galois flag code has an underlying $\beta$-Galois flag code. Thus, we have addressed the question of determine if the parameters and the behaviour of Galois flag codes drives, in some sense, the ones of the generalized ones. To do this, we provide a systematic construction of generalized Galois flag codes with a prescribed underlying Galois flag code that presents remarkable properties and helps to us to shed some light on the raised questions.
	
	To future work, we want to deepen the study of $\beta$-cyclic orbit codes by determining suitable generating flags that allow us to obtain a prefixed distance value and code sizes as large as possible, even when it is necessary to take unions or orbits.

\end{document}